\documentclass[11pt]{article}
\usepackage{fullpage}
\usepackage{amsmath}[1999/12/15] 
\usepackage{amssymb}
\usepackage{epic,eepic}

\newtheorem{definition}{Definition}[section]

\newtheorem{theorem}[definition]{Theorem}
\newtheorem{facten}[definition]{Fact}

\newtheorem{corollary}[definition]{Corollary}
\newtheorem{claim}[definition]{Claim}
\newtheorem{obs}[definition]{Observation}
\newtheorem{notat}[definition]{Notation}
\newtheorem{lemma}[definition]{Lemma}
\newtheorem{remark}[definition]{Remark}
\newtheorem{example}[definition]{Example}
\newcommand{\qed}{\hspace*{\fill} $\Box$}

\newenvironment{notation}{\begin{notat} \sl}{\end{notat}}
\newenvironment{proof}{\noindent {\bf Proof:} \hspace{.677em}}{\qed}

\newcommand{\id}{\text{Id}}
\newcommand{\su}{\subseteq}

\newcommand{\pr}{\Pr}

\newcommand{\sgn}{\operatorname {sgn}}

\newcommand{\DIST}{\operatorname{DIST}}
\newcommand{\SUM}{\operatorname{SUM}}
\newcommand{\GAP}{\operatorname{GAP-TR}}

\def\myvec#1{{\bf #1}}
\newcommand{\vecy}{{\myvec y}}
\newcommand{\vecz}{{\myvec z}}

\def\eps{{\varepsilon}}

\newcommand{\comment}[1]{}

\def\A{{\cal A}}

\def\C{{\cal C}}
\def\D{{\cal D}}

\def\I{{\cal I}}

\def\N{{\mathbb N}}

\def\P{{\cal P}}

\def\R{{\mathbb R}}

\def\T{{\cal T}}

\def\V{{\cal V}}

\newcommand{\eqdef}{\stackrel{\Delta}{=}}
\newcommand{\expected}{\mathbb E}

\newcommand{\corref}[1]{Corollary~\ref{cor:#1}}
\newcommand{\defref}[1]{Definition~\ref{def:#1}}
\newcommand{\lemref}[1]{Lemma~\ref{lem:#1}}

\newcommand{\secref}[1]{Section~\ref{sec:#1}}

\newcommand{\thmref}[1]{Theorem~\ref{thm:#1}}
\newcommand{\clmref}[1]{Claim~\ref{clm:#1}}
\newcommand{\exampleref}[1]{Example~\ref{ex:#1}}
\newcommand{\eqnref}[1]{Equation~(\ref{eqn:#1})}
\newcommand{\stpref}[1]{Step~(\ref{stp:#1})}
\newcommand{\notateref}[1]{Notation~\ref{notate:#1}}
\newcommand{\figref}[1]{Figure~\ref{fig:#1}}

\newcommand{\set}[1]{\left\{#1\right\}}

\newcommand{\remove}[1]{ \newpage \\ REMOVED TEXT \\ #1 \newpage}
\newcommand{\authornote}[2]{\noindent{\bf ({\rm\bf #1:} {\it #2})}}

\newcommand{\enote}[1]{\authornote{ERAN}{\textcolor{red}{#1}}\\}

\newcommand{\vect}[1]{(#1)}

\newcommand{\expon}[1]{\exp\left(#1\right)}
\newcommand{\View}{{\rm View}}
\newcommand{\Output}{{\rm Output}}
\newcommand{\Lap}{\mbox{\rm Lap}}
\newcommand{\GS}{\mbox{\rm GS}}
\newcommand{\RR}{{\mathbb R}}
\newcommand{\Var}{\mbox{\rm Var}}
\newcommand{\Exp}{{\mathbb E}}
\newcommand{\Flip}{\mbox{\rm flip}}

\begin{document}

\begin{titlepage}
\title{Distributed Private Data Analysis: \\ On Simultaneously Solving {\em How} and {\em What}\thanks{A preliminary version of this work appeared in David Wagner editor, Advances in Cryptology -- CRYPTO 2008. Volume~5157 of Lecture Notes in Computer Science, pages~451--468. Springer, 2008.} }

\author{{Amos Beimel \quad Kobbi Nissim \quad Eran Omri} \\ 
  Department of Computer Science \\ Ben Gurion University \\ Be'er Sheva, Israel \\
  {\tt \{beimel, kobbi, omrier\}@cs.bgu.ac.il}}


\maketitle

\begin{abstract} 

We examine the combination of two directions in the field of privacy concerning computations over distributed private inputs -- {\em secure function evaluation} (SFE) and {\em differential privacy}. While in both the goal is to privately evaluate some function of the individual inputs, the privacy requirements are significantly different.
The general feasibility results for SFE suggest a natural paradigm for implementing differentially private analyses distributively: First choose {\em what} to compute, i.e., a differentially private analysis; Then decide {\em how} to compute it, i.e., construct an SFE protocol for this analysis.

We initiate an examination whether there are advantages to a paradigm where both decisions are made simultaneously. In particular, we investigate under which accuracy requirements it is beneficial to adapt this paradigm for computing a collection of functions including binary sum, gap threshold, and approximate median queries. Our results imply that when computing the binary sum of $n$ distributed inputs then:
\begin{itemize}
\item
When we require that the error is  $o(\sqrt{n})$ and the number of rounds
is constant, there is no benefit in the new paradigm.
\item
When we allow an error of $O(\sqrt{n})$, the new paradigm yields more efficient protocols when we consider protocols that compute symmetric functions.
\end{itemize}
Our results also yield new separations between the local and global models of computations for private data analysis.
\end{abstract}

\paragraph{Keywords.} Differential privacy, Secure Function Evaluation, Sum Queries.
\thispagestyle{empty}
\end{titlepage}

\section{Introduction}
\label{sec:intro}

We examine the combination of two directions in the field of privacy concerning distributed private inputs -- secure function evaluation~\cite{Yao82a,GMW87,CCD88,BGW88} and differential privacy~\cite{DMNS06,Dwo06}. While in both the goal is to privately evaluate some function of individual inputs, the privacy requirements are significantly different.

Secure function evaluation (SFE) allows $n$ parties $p_1,\ldots,p_n$, sharing a common interest in distributively computing a function $f(\cdot)$ of their inputs ${\myvec x}=(x_1,\ldots,x_n)$, to compute $f({\myvec x})$ while making sure that no coalition of $t$ or less curious parties learns more than the outcome of $f({\myvec x})$. I.e., for every such coalition, executing the SFE protocol is equivalent to communicating with a trusted party that is given the private inputs ${\myvec x}$ and releases $f({\myvec x})$.
SFE has been the subject of extensive cryptographic research (initiated in~\cite{Yao82a,GMW87,CCD88,BGW88}), and SFE protocols exist for any feasible function $f(\cdot)$ in a variety of general settings.

SFE is an important tool for achieving privacy of individual entries -- no information about these entries is leaked beyond the outcome $f({\myvec x})$. However this guarantee is insufficient in many applications, and care must be taken in choosing the function $f(\cdot)$ to be computed -- any implementation, no matter how secure, of a function $f(\cdot)$ that leaks individual information would not preserve individual privacy. 

A criterion for functions that preserve privacy of individual entries, {\em differential privacy}, has evolved in a sequence of recent works~\cite{DiNi03,EGS03,DwNi04,BDMN05,DMNS06,Dwo06,DKMMN06}. It has been  demonstrated that differentially private analyses exist for a variety of tasks including the approximation of numerical functions (by adding carefully chosen random noise that conceals any single individual's contribution)~\cite{DMNS06,BDMN05,NRS07,GRS09}, non-numerical analyses~\cite{MT07}, datamining~\cite{BDMN05,NRS07}, learning~\cite{BDMN05,KLNRS08}, non-interactive sanitization~\cite{BLR08,DNRRS09,FKKN09}, and statistical analysis~\cite{DL09,Smi08}.

Employing the generality of SFE, we can combine these to directions in a natural paradigm for constructing protocols in which differential privacy is preserved:
\begin{enumerate}
\item Decide on {\em what} to compute. This can be, e.g., a differentially private analysis $\hat{f}(\cdot)$ that approximates a desired analysis $f(\cdot)$. Designing $\hat{f}(\cdot)$ can be done while abstracting out all implementation issues, assuming the computation is performed by a trusted party that only announces the outcome of the analysis.
\item Decide on {\em how} to compute, i.e., construct an SFE protocol for computing $\hat{f}({\myvec x})$ either by using one of the generic transformations of the feasibility results mentioned above, or by crafting an efficient protocol that utilizes the properties of $\hat{f}(\cdot)$. 
\end{enumerate}

This natural paradigm yields a conceptually simple recipe for constructing
distributed analyses preserving differential privacy, and, furthermore,
allows a valuable separation of our examinations of the {\em what} and {\em
how} questions. 

Comparing the privacy requirements of SFE
protocols with differential privacy suggests, however, that this combination may
result in sub-optimal protocols. For example, differential privacy is only
concerned with how the view of a coalition changes when one (or only few)
of the inputs are changed, whereas SFE protocols are required to keep these
views indistinguishable even when significant changes occur, if these
changes do not affect the computed function's outcome. Hence, it may be advantageous to consider a paradigm where the analysis to be 
computed and the protocol for computing it are chosen simultaneously.

\subsection{Our Underlying Models}

The main distributed model we consider is of $n$ honest-but-curious (a.k.a.\ semi-honest) parties $p_1,\ldots,p_n$ that are connected via a complete network and perform a computation over their private inputs $x_1,\ldots,x_n$. Privacy is required to be maintained with respect to all coalitions of size up to $t$. The model of honest-but-curious parties has been examined thoroughly in cryptography, and was shown to enable SFE in a variety of settings~\cite{Yao82a,GMW87,BGW88,CCD88}. We change the standard definition so that differential privacy has to be maintained with respect to coalitions of curious parties (see \defref{dist-dp} below). 

Another distributed model we consider is the {\em local model}%
\footnote{Also referred to in the literature as {\em randomized response} and {\em input perturbation}. 
This model was originally introduced by Warner~\cite{W65} as a means of encouraging survey responders to 
answer truthfully, and has been studied extensively since.}.
Protocols executing in the local model have a very simple communication structure, where each party $p_i$ can only communicate with a designated honest-but-curious party $C$, which we refer to as a {\em curator}. The communication can either be {\em non-interactive}, where each party sends a single message to the curator which replies with the protocol's outcome, or {\em interactive}, where several rounds of communication may take place. 

While it is probably most natural to consider a setting where the players are computationally limited (i.e., all are probabilistic polynomial time machines), we present our results in an information theoretic setting. This choice has two benefits:
\begin{itemize}
\item Technically, it allows us to prove lower bounds on SFE protocols (where similar bounds are not known for the computational setting). Hence, we can rigorously demonstrate when constructing differentially private protocols is better than using the natural paradigm. 
\item On the flip side, our bounds on the information theoretic model demonstrate, for the first time, a setting where reliance on computational hardness assumptions strictly improves the construction of differentially private analyses. 
\end{itemize}

\subsection{Our Results}

We initiate an examination of the paradigm where an analysis and the protocol for computing it are chosen simultaneously. We begin with two examples that present the potential benefits of using this paradigm: it can lead to simpler protocols, and more importantly it can lead to more efficient protocols. For the latter we consider the Binary Sum function, 
$$\SUM(x_1,\ldots,x_n) = \sum_{i=1}^n x_i\;\;\; \mbox{for}~x_i\in\{0,1\}.$$

The major part of this work examines whether constructing non-SFE protocols for
computing an approximation $\hat{f}(\cdot)$ to $\SUM(\cdot)$ yields an efficiency gain%
\footnote{We only consider {\em oblivious protocols} where the communication pattern is 
independent of input and randomness (see \secref{dpp}).}.  
Ignoring the dependency on the privacy parameter, our first observation is that for approximations 
with additive error $\approx \sqrt{n}$ there is a gain -- for a natural class of {\em symmetric}
approximation functions (informally, functions where the outcome does not
depend on the order of inputs), it is possible to construct differentially
private protocols that are much more efficient than any SFE protocol for a
function in this class.  Moreover, these differentially private protocols
are secure against coalitions of size up to $t=n-1$, and need not rely on
secure channels.

The picture changes when we consider additive error smaller than $\sqrt n$. 
This follows from a sequence of results:
\begin{enumerate}
\item We prove first that no such non-interactive protocols in the local model
  exist.
      Furthermore, no local protocols with $\ell \leq \sqrt{n}$ rounds and additive error 
      $\sqrt{n}/\tilde{O}(\ell)$ exist. 
\item We show that in particular, no local interactive protocol with 
      $o(\sqrt{n/\log n })$ rounds exists for computing $\SUM(\cdot)$ within constant additive 
      error (this is in contrast to the centralized setup where $\SUM(\cdot)$ can be 
      computed within $O(1)$ additive error).
\item Finally, we prove that the bounds on local protocols imply that no
                        distributed protocols exist that use $nt/4$ messages, and approximates $\SUM(\cdot)$ within additive 
                        error $\sqrt{n}/\tilde{O}(\ell)$ in $\ell$ rounds. 
\end{enumerate}
Considering the natural paradigm, i.e., computing a
differentially-private approximation to $\SUM(\cdot)$ using SFE, we get a protocol for
approximating $\SUM(\cdot)$ with $O(1)$ additive error, and sending $O(nt)$ messages. 
Thus, for protocols with error $o(\sqrt{n}/\eps)$ and small number of rounds, there 
is no gain in using the paradigm of a simultaneous design of the function and its protocol.

Our results imply that differentially private protocols constructed under computational hardness assumptions, yielding a computational version of differential privacy (see Definition~\ref{def:comp-dist-dp}), are provably more efficient than protocols that do not make use of computational hardness. 
For instance, the phase transition we observe at $\theta(\sqrt{n}/\eps)$ additive error {\em does not} hold in a computational setting. See \exampleref{comp} for details.

\subsection{Techniques}
We prove our lowerbound for the distributed model in a sequence of reductions. We begin with a simple reduction from any differentially private protocol for 
$\SUM$ to a gap version of the threshold function, denoted $\GAP$. Henceforth, it is enough to prove our lowerbound for $\GAP$. 

In the heart of our lowerbound for $\GAP$ is a transformation from efficient
distributed protocols into local interactive protocols, showing that if
there are distributed differentially-private protocols for $\GAP(\cdot)$ in which half of the parties interact with less than
$t+1$ parties, then there exist differentially-private protocols for $\GAP(\cdot)$ in the local interactive model. This allows us
to prove our impossibility results in the local model, which is considerably simpler to analyze. 

In analyzing the local non-interactive model, we prove lowerbounds borrowing from analyses in~\cite{DiNi03,DwNi04}. The main technical difference is that our analysis is a lowerbound and hence should hold for general protocols, whereas the work in~\cite{DiNi03,DwNi04} was concerned with proving feasibility of privacy preserving computations (i.e., upperbounds), and hence they analyze of very specific protocols.

To extend our lowerbounds from the local non-interactive  to
interactive protocols, we decompose an $\ell$-round interactive protocol to
$\ell$ one-round protocols, analyze the $\ell$ protocols,
and  use composition to obtain the lowerbound.

\subsection{Related Work}

Secure function evaluation and private data analysis were first tied together in the {\em Our Data, Ourselves (ODO)} 
protocols~\cite{DKMMN06}. The constructions in~\cite{DKMMN06}
-- distributed SFE protocols for generating shares of random noise 
used in private data analyses -- follow the natural paradigm discussed above (however, they avoid utilizing 
generic SFE feasibility results to gain on efficiency). We note that a  difference between the protocols 
in~\cite{DKMMN06} and the discussion herein is that ODO protocols are secure against malicious parties, in a 
computational setup, whereas we deal with honest-but-curious parties, and mostly in an information theoretic setup.
Following our work, computational differential privacy was considered in~\cite{MPRV09}; they present several definitions of computational differential privacy,
study the relationships between these definitions, and construct efficient 2-party computational differentially private protocols for approximating the distance between two vectors. In this work, we supply a definition of computationally
$(t,\epsilon)$-differentially  private protocols which is close to the definition of IND-CDP privacy in~\cite{MPRV09}. 

Lowerbounds on the local non-interactive model were previously presented
implicitly in~\cite{DMNS06,RHS07,KLNRS08}, and explicitly
in~\cite{DiNi03,DMT07}. The two latter works are mainly concerned with what
is called the global (or centralized) interactive setup, but have also
implications to approximation to $\SUM$ in the local {\em non-interactive}
model, namely, that it is impossible to approximate it within additive
error $c\sqrt{n}$ (for some constant $c>0$), a slightly weaker result
compared to our lowerbound of $c\sqrt{n}/\eps$ for
$\eps$-differentially private local non-interactive protocols. However,
(to the best of our understanding) these implications
of~\cite{DiNi03,DMT07} do not imply the lowerbounds we get for local
interactive protocols and distributed protocols.

Chor and Kushilevitz~\cite{CK93} consider the problem of securely computing
modular sum when the inputs are distributed. They show that this task can
be done while sending roughly $n(t+1)/2$ messages. Furthermore, they prove
that this number of messages is optimal for a family of protocols that they
call oblivious. These are protocols where the communication pattern is fixed
and does not depend on the inputs or random inputs. In our work we extend their
lowerbound result and prove that with $n(t+1)/4$ messages no symmetric approximation
for $\SUM$ with sub-linear additive error can be computed in an oblivious protocol. 

\subsection{Organization}

The rest of the paper is organized as follows: In \secref{preliminaries} we
define differentially private analyses and its extension to differentially
private protocols (both information-theoretic and computational), describe the local model of communication, and define the
binary sum and gap threshold functions.  In \secref{motivatingExamples}, we
present two motivating examples for our new methodology of simultaneously
solving how and what. 
In \secref{local} we prove lowerbounds on the error of differentially
private protocols for computing the
binary sum and gap threshold functions in the local model, and in
\secref{dist} we extend these lowerbounds to the distributed model.
Finally, in \secref{SymmetricFunctions} we prove that an SFE protocol for computing a symmetric approximation of the sum function with less than
$nt/4$ messages has an error of $\Omega(n)$ (compared to a non-SFE protocol
that approximates the sum function with 
$O(n)$ messages and an error of $\Omega(\sqrt{n})$).

\section{Preliminaries}
\label{sec:preliminaries}

\paragraph{Notation.}  A {\em vector} ${\myvec x}=(x_1,\ldots,x_n)$ is an ordered sequence of $n$ elements of some domain $D$.  
Vectors $\myvec{x},\myvec{x'}$ are {\em neighboring} if they differ on exactly one entry, and are {\em $T$-neighboring} if they differ on a single entry whose index is {\em not} in $T\subset [n]$.

The {\em Laplace distribution}, $\Lap(\lambda)$, is the continuous probability distribution with probability 
density function $$h(y) =\frac{\exp(-|y|/\lambda)}{2\lambda}.$$  
For $Y\sim\Lap(\lambda)$ we have that $\expected[Y]=0$, $\Var[Y]=2\lambda^2$, and $\Pr[|Y|> k\lambda] = e^{-k}$.

\begin{definition} \label{def:randomizedFunction}
        Let $D_I$, $D_R$, and $R$ be sets. 
        An $n$-ary {\em randomized function} is a function $\hat{f}: \left (D_I\right)^n  \times D_R \rightarrow R$, 
        where $D = D_I$ is the domain of $\hat {f}$ and $D_R$ is the set of random inputs. 
        For $\myvec {x} = \vect{x_1, \ldots, x_n} \in D^n$ we usually write $\hat{f}(\myvec x)$ with the underlying 
convention that $\hat{f}(x_1, \ldots, x_n) = \hat{f}(x_1, \ldots, x_n,r)$, where  $r$
        is  uniformly selected form $D_R$. Following this convention, we also usually omit $D_R$
        from the notation and write $\hat{f}: D^n \rightarrow R$.
\end{definition}

\subsection{Differential Privacy} 
\label{sec:dp}
Our privacy definition for distributed protocols (\defref{dist-dp} below) can be viewed
as a distributed variant of $\eps$-differential privacy. Informally, a computation is
differentially private if any change in a single individual input may only
induce a small change in the distribution on its outcomes.

\begin{definition}[Differential privacy~\cite{DMNS06}]\label{def:dp}
Let $\hat{f} : \D^n \rightarrow R$ be a randomized function from domain
$D^n$ to range $R$. We say that $\hat{f}$ is {\em
$\eps$-differentially private} if for all neighboring vectors ${\myvec
x},{\myvec x'}$, and for all possible sets of outcomes $\V \subseteq R$ it
holds that
$$\Pr[\hat{f}({\myvec x}) \in \V] \leq e^\eps\cdot \Pr[\hat{f}({\myvec x'}) \in \V].$$
The probability is taken over the randomness of $\hat{f}$.
\end{definition}

Several frameworks for constructing differentially private functions by
means of perturbation are presented in the literature
(see~\cite{DMNS06,BDMN05,NRS07,MT07}). 
The most basic transformation on a function $f$ that yields a
differentially private function  is via the framework of {\em global
sensitivity}~\cite{DMNS06}. In this framework the outcome is obtained by adding to $f({\myvec x})$ noise sampled from the Laplace 
distribution, calibrated to the global sensitivity of $f$, defined as 
$$\GS_f = \max |f({\myvec x})-f({\myvec x'})|\text{, with the maximum taken over neighboring }{\myvec x},{\myvec x'}.$$
Formally, $\hat{f}$ is defined as
\begin{equation}
\label{eqn:gsmech}
\hat{f}({\myvec x}) = f({\myvec x}) + Y\text{, where }Y\sim\Lap(\GS_f/\eps).
\end{equation}

\begin{example}
\label{ex:binary-sum}
The {\em binary sum} function $\SUM: \{0,1\}^n \rightarrow \RR$ is defined as $\SUM({\myvec x}) = \sum_{i=1}^n x_i$. 
For every two neighboring ${\myvec x},{\myvec x'} \in \{0,1\}^n$ we have that 
$|\SUM({\myvec x}) - \SUM({\myvec x'})| = 1$ and hence $\GS_{\SUM} = 1$. 
Applying \eqnref{gsmech}, we get an  $\eps$-differentially private approximation, 
$ \hat{f}({\myvec x}) = \SUM({\myvec x}) + Y \mbox{\rm, where\ } Y \sim \Lap(1/\eps)$, that is,
we get a differentially private approximation of \/ $\SUM$ with $O(1)$ additive error.
\end{example}

\subsection{Differentially Private Protocols}
\label{sec:dpp}

We consider a distributed setting, where $n$ parties $p_1,\dots,p_n$ hold private inputs $x_1,\ldots,x_n$ respectively and engage in a protocol $\Pi$ in order to compute (or approximate) a function $f(\cdot)$
of their joint inputs. 
Parties are {\em  honest-but-curious}, which means they follow the prescribed
randomized protocol. However, as the execution of the protocol terminates, 
colluding parties can try to infer information about inputs of parties 
outside the coalition, given their joint view of the execution.

The protocol $\Pi$ is executed in a synchronous environment with point-to-point secure (untappable) communication channels, and is required to preserve privacy with respect to coalitions of up to $t$ parties. Following~\cite{CK93}, we assume that the protocol $\Pi$ has a {\em fixed-communication} pattern (such protocols are called {\em oblivious}), i.e., every channel is either (i) active in every run of $\Pi$ (i.e., at least one bit is sent over the channel), or (ii) never used%
\footnote{Our proofs also work in a relaxed setting where every channel is
either (i) used in at least a constant fraction of the runs of $\Pi$ (where
the probability is taken over the coins of $\Pi$), or (ii) is never
used.}.
Parties that are adjacent to at least $t+1$  active channels are called {\em popular} other parties are called {\em lonely}.

The main definition we will work with is an extension of \defref{dp} to a distributed setting. 
Informally, we require that differential privacy is preserved with respect to any coalition of size up to $t$.

\begin{definition}[Distributed differential privacy] \label{def:dist-dp}
Let $\Pi$ be a protocol between $n$ (honest-but-curious) parties. For a set $T\subseteq [n]$
and fixed inputs $\myvec{x} = \vect{x_1,\ldots,x_n}$, let
$\View_T\vect{x_1,\ldots,x_n}$ be the random variable containing the inputs of the parties in $T$ (i.e., $\set{x_i}_{i\in T}$), the random inputs of the parties in $T$, and the messages that the parties
in $T$ received during the execution of the protocol with private inputs ${\myvec x} = \vect{x_1,\ldots,x_n}$
(the randomness is taken over the random inputs of the parties not in $T$). 

We say that $\Pi$ is {\em $(t,\eps)$-differentially private} if for all $T \subset [n]$, where $|T| \le t$, 
for all $T$-neighboring ${\myvec x},{\myvec x'}$, and for all possible sets $\V_T$ of views of the parties in $T$:
\begin{equation}
\Pr[\View_T({\myvec x}) \in \V_T] \leq e^\eps\cdot \Pr[\View_T({\myvec x'}) \in \V_T],
\label{eqn:dist-dp}
\end{equation}
where the probabilities are taken over the random inputs of the parties in the
protocol $\Pi$.
\end{definition}

An equivalent requirement is that for all $T \subset [n]$, where $|T| \le t$, 
for all $T$-neighboring ${\myvec x},{\myvec x'}$, and for all distinguishers $D$ (i.e., functions, not necessarily efficiently computable, from views to $\{0,1\}$),
$$\Pr[D(\View_T({\myvec x})) =1] \leq e^\eps\cdot \Pr[D(\View_T({\myvec x'})) =1 ].$$ 
This requirement can be relaxed to only consider distinguishers that are computationally bounded:

\begin{definition}[Computational distributed differential privacy]
\label{def:comp-dist-dp} 
We say that $\Pi$ is {\em computationally $(t,\eps)$-differentially  private} if for every probabilistic 
polynomial-time algorithm $D$, and for every polynomial $p(\cdot)$, there exists $k_0$ such that for all 
$k \ge k_0$, for all $T \subset [n]$, where $|T| \le t$, and for all $T$-neighboring inputs
$\myvec {x},\myvec {x'}\in \vect{\set{0,1}^k}^n${\upshape :}
$$\Pr[D(\View_T({\myvec x})) =1]  \leq  e^\eps\cdot \Pr[D(\View_T({\myvec x'})) =1 ] + \frac{1}{p(n\cdot k)}\;,$$
where the probabilities are taken over the random inputs of the parties in
protocol $\Pi$ and the randomness of $D$.
\end{definition}

\begin{example}
\label{ex:comp}
We next describe a computationally $(n/2,\eps)$-differentially  private protocol for computing $\SUM$ with $O(\log n/\epsilon)$ additive error, $O(n)$ messages, and constant number of rounds.
This protocol uses a homomorphic encryption scheme with threshold decryption (that is, only the sets of all parties can decrypt messages). For example, if we use ElGamal encryption, the distributed key generation and decryption require one round in which each party sends one message. 
The protocol works in three phases:
\begin{description}
\item[Key Generation.]  The parties generate public and private keys for the homomorphic encryption scheme with threshold decryption.                        
\item[Encryption.]
  Each party $p_i$ chooses a random ${\rm noise}_i$ (according to a distribution
  that will be defined  later), computes $y_i=x_i+{\rm noise}_i$, encrypts $y_i$ using the public encryption key and sends the encryption  to $p_1$.  
\item[Decryption.]
Party $p_1$ computes $z$, an  encryption of $y=\sum_{i=1^n} y_i$
(this is possible as the encryption scheme is homomorphic).
$p_1$ sends $z$ to each $p_i$, which in return sends a decryption message back
to $p_1$.
Finally, $p_1$ decrypts $y$ from the decryption messages and sends $y$ to all
parties.
\end{description}
One way to generate each party's noise is for each party to sample from the
Normal distribution with mean zero and variance
$6\log^2n/(n\eps^2)$. Since the sum of normal random variables is a
normal random variable, $y=\sum_{i=1^n} x_i +{\rm noise}$ where noise is sampled from a normal distribution with mean zero and variance
$6\log^2n/\eps^2$. Furthermore, even if a coalition of $n/2$ parties
subtracts the noise that its parties added to $y$, the variance of the remaining noise is
$3\log^2n/\eps^2$. Using the analysis of~\cite{DKMMN06}, the protocol
is a computationally $(n/2,\eps)$-differentially private protocol
which with constant probability has error  $O(\log n/\epsilon)$.

The above protocol is a computationally $(n/2,\eps)$-differentially private protocol for computing $\SUM$ with $O(\log n/\epsilon)$ additive error, $O(n)$ messages, and constant number of rounds. In contrast, we prove that $(n/2,\eps)$-differentially information-theoretically private protocol for computing $\SUM$ with $o(\sqrt{n})$ additive error and constant number of rounds must send $\Omega(n^2)$ messages. Thus, our results shows that requiring only
computational differentially-privacy
does result in  more efficient protocols.
\end{example}

Using standard SFE feasibility results (in the computational setting), it is possible now to prove that the natural paradigm presented in \secref{intro} yields protocols that adhere to \defref{comp-dist-dp}. 
Consider an $\epsilon$-differentially private data analysis $\hat{f}$ and a computationally bounded distinguisher $D$, trying distinguish between a computation of an SFE protocol computing $\hat{f}$ with neighboring inputs $\myvec {x}$ and $\myvec {x'}$. Since, $\hat{f}$ preserves differential privacy the distributions on the outputs must be $\eps$ close, the same must hold for the random variables describing
the adversary's view (up to some negligible function in the length of the (concatenated) inputs). We get:
\begin{lemma}[Informal]\label{lem:dp+sfe=dpsfe}
Let $\hat{f}$ be $\eps$-differentially private, and let $\Pi$ be a $t$-secure protocol computing $\hat{f}$, then $\Pi$ is computationally $(t,\eps)$-differentially private.
\end{lemma}
In the above lemma, the  if the $t$-secure protocol $\Pi$ computing $\hat{f}$ has perfect security, then $\Pi$ is information-theoretically $(t,\eps)$-differentially private.

\begin{remark}
We will only consider protocols computing a (randomized) function $\hat{f}(\cdot)$ resulting in all 
parties computing the {\em same} outcome of $\hat{f}({\myvec x})$. This can be achieved, e.g., by 
having one party compute $\hat{f}({\myvec x})$ and send the outcome to all other parties.
\end{remark}

\subsection{Distributed Protocols -- Basic Observations}

The following notation and basic observations are used throughout the paper.

\begin{notation}
\label{notate:trnascriptProbProvidedMsgs}
Fix an $n$-party randomized protocol $\Pi$ and fix some communication transcript $c$. Assume that party $p_i$ holds an input $x_i$ and receives messages according to the transcript $c$. We define $\alpha_i^{c}(x_i)$ to be the probability that on input $x_i$ party $p_i$ sends messages that are consistent with transcript $c$, given that it receives messages that are consistent with $c$.
The probability is taken over the randomness of party $p_i$.
\end{notation}

Let $\ell$ be the number of rounds in $\Pi$. Assume, without loss of generality, that $p_i$ receives and sends messages in every round, and let $\beta^c_j$ (where $1\leq j\leq \ell$) be the probability on input $x_i$ party $p_i$ sends in round $j$ messages that are consistent with transcript $c$, provided that in previous rounds $p_i$ sees messages that are consistent with $c$. Then, by the chain rule of conditional probabilities we have that 
$$\alpha_i^{c}(x_i) = \prod_{j=1}^{\ell}\beta^c_j.$$
%
%
Observe that the event that $p_i$ sends messages according to $c$ when it sees messages according to $c$ depends only on the randomness $r_i$, and hence this event is independent of whether the other parties  send messages according to $c$ when they see messages according to $c$. We hence get the following lemma:

\begin{lemma}
\label{lem:mult}
Fix an $n$-party randomized protocol $\Pi$, assume that each $p_i$ holds an
input $x_i$, and fix some communication transcript $c$.  Then, the
probability that $c$ is exchanged is $\prod_{i=1}^{n}\alpha_i^{c}(x_i)$.
\end{lemma}

\subsection{The Local Model} 
\label{sec:local-model}

The local model (previously discussed in~\cite{DMNS06,KLNRS08}) is a
simplified distributed communication  model where the parties communicate
via a designated party -- a {\em curator%
\footnote{Unlike in a centralized setting where the curator is a trusted party that collects raw private information, in the local model the curator is a non-trusted party. In our setting, the curator is semi-honest.}} -- denoted $C$. The curator has no local input.  We will consider two types of differentially private local
protocols -- interactive and non-interactive.

In {\em non-interactive} local protocols each  party $p_i$
applies an $\eps$-differentially private algorithm $S_i$ on its private
input $x_i$ and randomness  $r_i$, and sends $S_i(x_i, r_i)$ to $C$ that
then performs an arbitrary computation and publishes its result.

In {\em interactive} local protocols the protocol proceeds in {\em rounds},
where in each round $j$ the curator sends to each party $p_i$ a ``query''
message $q_{i,j}$ and party $p_i$ responds with the $j$th ``answer'' 
$A_i(x_i, q_{i,1},\ldots,q_{i,j}, r_i)$; the answer is a function of the party's
input $x_i$, its random input $r_i$, and the first $j$ queries. I.e., each
round consists of two communication phases: first, the query messages are sent
by the curator, then, each party sends the appropriate response message.

We note that in the honest-but-curious setting we can assume, without loss of
generality, that the curator is deterministic, as randomness for the
curator may be provided by parties in their first message.

\begin{definition}[Differential privacy in the local model]
\label{def:local-model-privacy-collective}
We say that a protocol $\Pi$ in the local model is {\em $\eps$-differentially private} if the curator's view preserves $\eps$-differential privacy. Formally, for all neighboring ${\myvec x},{\myvec x'}$ and for every possible set $\V_C$ of views of the curator:
$$ \Pr[\View_{C}({\myvec x}) \in \V_C] \leq e^\eps\cdot \Pr[\View_{C}({\myvec x'}) \in \V_C], $$
where $\View_{C}({\myvec x})$ is the random variable containing the messages that $C$ receives  during the execution of the protocol with private inputs ${\myvec x} = (x_1,\ldots,x_n)$ and the probability is taken over the random inputs of the parties.
\end{definition}

We note that $\View_{C}({\myvec x})$ is defined in accordance with \defref{dist-dp} (with some abuse of notation, as we write $C$ instead of $\{C\}$). 
However, since  $C$ has no initial input and since $C$ is assumed to be deterministic, it is enough to include in $\View_{C}({\myvec x})$ only the messages that $C$ receives during the execution of the protocol with inputs ${\myvec x} = (x_1,\ldots,x_n)$.

Differential privacy in the local model may be equivalently phrased as a requirement to preserve  the privacy of each party independently of other parties. We next give a definition in this spirit by
considering the probabilities that a party $p_i$ replies in a certain way
to a given sequence of queries with, say, $x_i = 0$ and with, say, $x_i = 1$.
Any communication transcript $c$ in an execution of the protocol
defines a transcript $c_i$, where $$c_i = (q_{i,1},a_{i,1},\ldots,q_{i,\ell},a_{i,\ell})$$ is the restriction of $c$ to
the messages transferred between party $p_i$ and the curator (recall that
in the local model every party communicates solely with the curator). Thus,
we can use $\alpha_i^{c_i}(x_i)$ (see
\notateref{trnascriptProbProvidedMsgs}) to denote the probability that
$p_i$ with private input $x_i$ replies by $a_{i,1},\ldots,a_{i,\ell}$
provided the curator has sent queries $q_{i,1},\ldots,q_{i,\ell}$. Using this notation, we
present the alternative definition of privacy in the 
local model.
\begin{definition}[Differential privacy in the local model by individual privacy]
\label{def:local-model-privacy-individual}
We say that a protocol $\Pi$ in the local model is {\em
$\eps$-differentially private} if  the curator's view preserves
$\eps$-differential privacy with respect to each party separately.
Formally, for every $i \in [n]$ and for any possible communication
transcript $c_i = (q_{i,1},a_{i,1},\ldots,q_{i,\ell},a_{i,\ell})$ between
party $p_i$ and the curator (i.e., there exist inputs $x'_1,\dots,
x'_n$ and random inputs $r'_1,\dots, r'_n$ consistent with $c_i$), and for every
$x_i,y_i \in D$ it holds that  $$\alpha_i^{c_i}(x_i) \le e^\eps
\cdot \alpha_i^{c_i}(y_i),$$ where the probabilities are taken over the random
input of $p_i$.
\end{definition}

\begin{claim} 
\label{clm:local-model-def-equivalence}
\defref{local-model-privacy-collective} is equivalent to 
\defref{local-model-privacy-individual}.
\end{claim}

\begin{proof}
We prove implications in both directions.

\paragraph{\defref{local-model-privacy-collective} $\Rightarrow$ \defref{local-model-privacy-individual}:}
Let $\Pi$ be according to \defref{local-model-privacy-collective}.
Given a possible transcript $c_i$ of messages between party $p_i$ and $C$, choose any possible transcript $c = (c_1,\ldots,c_n)$ that is consistent with $c_i$. 
We get that for all $x_i,y_i$,
$$\frac{\alpha_i^{c_i}(x_i)}{\alpha_i^{c_i}(y_i)} =  \frac{\alpha_i^{c_i}(x_i)}{\alpha_i^{c_i}(y_i)} \cdot \frac{\prod_{j\not= i} \alpha_j^{c_j}(x_j)} {\prod_{j\not= i} \alpha_j^{c_j}(x_j)} = \frac{\Pr[\View_C(x_1,\ldots,x_{i-1},x_i,x_{i+1},\dots,x_n) = c]} {\Pr[\View_C(x_1,\ldots,x_{i-1},y_i,x_{i+1},\dots,x_n) = c]} 
\leq  e^\eps, $$
where the last equality follows by \lemref{mult}, and the last inequality follows from $\Pi$ being  $\eps$-differentially private according to \defref{local-model-privacy-collective}, noting that $(x_1,\ldots,x_{i-1},x_i,x_{i+1},x_n)$ and $(x_1,\ldots,x_{i-1},y_i,x_{i+1},x_n)$ are neighboring.

\paragraph{\defref{local-model-privacy-individual} $\Rightarrow$ \defref{local-model-privacy-collective}:}
Let $\pi$ be according to \defref{local-model-privacy-individual}. Given a possible transcript $c = (c_1,\ldots,c_n)$ and neighboring inputs $(x_1,\ldots,x_{i-1},x_i,x_{i+1},x_n)$ and $(x_1,\ldots,x_{i-1},y_i,x_{i+1},x_n)$, we have that 
$$\frac{\Pr[\View_C(x_1,\ldots,x_{i-1},x_i,x_{i+1},\dots,x_n) = c]} {\Pr[\View_C(x_1,\ldots,x_{i-1},y_i,x_{i+1},\dots,x_n) = c]} = 
\frac{\alpha_i^{c}(x_i)}{\alpha_i^{c}(y_i)} \cdot \frac{\prod_{j\not= i} \alpha_j^{c}(x_j)} {\prod_{j\not= i} \alpha_j^{c}(x_j)} = \frac{\alpha_i^{c_i}(x_i)}{\alpha_i^{c_i}(y_i)} 
\leq  e^\eps,$$
where the first equality follows by \lemref{mult} and the the inequality follows from $\Pi$ being $\eps$-differentially private according to \defref{local-model-privacy-individual}.
\end {proof}

\clmref{local-model-def-equivalence} implies that, in the information-theoretic local model, requiring differential privacy for the curator implies differential privacy with respect to every coalition.

\subsection{Approximation}
We will construct protocols whose outcome approximates a function $f:D^n
\rightarrow \RR$ by a probabilistic function, according to the following definition:

\begin{definition}[Approximation]
A randomized function $\hat{f}:D^n \rightarrow \RR$ is an {\em additive} $(\gamma,\tau)$-approximation for a (deterministic) function $f$ if $$\Pr\left[|f(\myvec{x})-\hat{f}(\myvec{x})| > \tau(n)\right] \leq \gamma(n)$$ 
for all $\myvec{x} \in D^n$. The probability is over the randomness of $\hat f$.
\end{definition}

For example, by the properties of the Laplace distribution, \eqnref{gsmech} yields 
an additive $(e^{-k},k\cdot\GS_f/\eps)$-approximation to $f$, for every $k>0$.

\subsection{The Binary Sum and Gap Threshold Functions}
\label{sec:sumandgap}

The {\em binary sum} function is defined to be
$\SUM_n(x_1,\ldots,x_n) = \sum_{i=1}^n x_i$ for $x_i\in\{0,1\}$ (the subscript $n$ is omitted when it is clear from the context). We will use a {\em gap} (or promise) version of the threshold function:
\begin{definition}[Gap Threshold] \label{def:gap}
For $\kappa,\tau >0$, 
$$\GAP_{\kappa,\tau}(x_1,\dots,x_n)= \left\{\begin{array}{ll}
        0 & \mbox{If}~\SUM_n(x_1,\dots,x_n)\leq \kappa, \\
        1 & \mbox{If}~\SUM_n(x_1,\dots,x_n)\geq \kappa+\tau.
\end{array}\right.$$
Note that $\GAP_{\kappa,\tau}(x_1,\dots,x_n)$ is not defined when $\kappa < \SUM_n(x_1,\dots,x_n)< \kappa+\tau$. 
\end{definition}

It is easy to transform any $(\gamma,\tau/2)$-approximation $\hat f$ of $\SUM$ to a $(\gamma,0)$-approximation $\hat g$ to $\GAP_{\kappa,\tau}$: given $y = \hat f (\myvec x)$ for $\SUM_{n}(\myvec x)$, set the $\hat g (\myvec x)$ to be $0$ if $y \leq \kappa + \tau/2$ and $1$ otherwise. We get the following simple corollary:

\begin{corollary} \label{cor:reduction-gap-sum}
If there exists an $\ell$-round, $(t,\eps)$-differentially private protocol
(resp.\ $\eps$-differentially private protocol in the local model) that $(\gamma,\tau/2)$-approximates $\SUM_n$ sending $\rho$ messages, then for every $\kappa$ there exists an $\ell$-round, $(t,\eps)$-differentially private protocol (resp.\ $\eps$-differentially private protocol in the local model) that correctly computes $\GAP_{\kappa,\tau}$ with probability at least $1-\gamma$, sending at most $\rho$ messages.
\end{corollary}

Specifically, non-existence of $(t,\eps)$-differentially private
protocols for computing $\GAP_{0,\tau}$ correctly with $n(t+1)/4$ messages
implies that there exists no $(t,\eps)$-differentially private
protocols for computing $\SUM_n$ with $n(t+1)/4$ messages and additive error magnitude $\tau/2$. 
The next claim asserts that the same non-existence also implies that, for any $0\leq \kappa\leq n-\tau$, 
there exists no $(t,\eps)$-differentially private protocol for computing $\GAP_{\kappa,\tau}$ 
correctly with $n(t+1)/8$ messages. Again, it applies to both the distributed and the local models.

\begin{claim} \label{clm:reduction-gapK-gap0}
If for some $0\leq \kappa\leq n-\tau$ there exists an $\ell$-round,
$(t,\eps)$-differentially private (respectively, $\eps$-differentially 
private in the local model) $n$-party
protocol that correctly computes $\GAP_{\kappa,\tau}$ with probability at
least $\gamma$ sending at most $\rho$ messages, then there exists an
$\ell$-round, $(t/2,\eps)$-differentially private (respectively, 
$\eps$-differentially private in the local model) $n/2$-party protocol
that correctly computes $\GAP_{0,\tau}$  with probability at least $\gamma$
sending at most $\rho$ messages.
\end{claim}

\begin {proof}
For $\kappa \leq n/2$, given an $n$-party  protocol $\Pi$ that correctly
computes  $\GAP_{\kappa,\tau}$, define an  $n/2$-party protocol $\Pi'$ 
for computing $\GAP_{0,\tau}$ by simulating parties
$p_{\frac{n}{2}+1}, \dots, p_n$ where  $x_{\frac{n}{2}+1}, \dots,
x_{\frac{n}{2}+\kappa}$ are set to $1$ and $x_{\frac{n}{2}+\kappa+1},
\dots, x_{n}$ are set to $0$. In the local model, a designated party,  say
$p_1$, can simulate these $n/2$ parties. In the distributed model,
we let each party $p_i$ simulate party $p_{i+n/2}$.

Observe that in the distributed model any view $v$ of a coalition $T'$ of size $t' \le t/2$ in some execution of the resulting protocol, is exactly the view of the coalition $T$ of size $2t' \le t$, implied by $T'$ (for $p_i \in T'$ we have $p_i,p_{i+n/2} \in T$), in the appropriate computation of the original protocol. Moreover, any $T'$-neighboring $\myvec{x},\myvec{x'}$ define $T$-neighboring $\myvec{x}\myvec{y},\myvec{x'}\myvec{y}$ 
(where $\myvec{y} = 1^{\kappa}0^{\frac{n}{2}-\kappa}$), such that $\Pr[\View_T({\myvec x}{\myvec y})  = v] = \Pr[\View_{T'}({\myvec x}) = v]$ and $\Pr[\View_T({\myvec x'}{\myvec y})  = v] = \Pr[\View_{T'}({\myvec x'}) = v]$. Thus, by the privacy of the original protocol, the resulting protocol is $(t/2,\eps)$-differentially private.

For $\kappa > n/2$, we can use the construction above to compute $\GAP_{n-\kappa-\tau,\tau}$, by flipping all input bits
(that is, changing 1 to 0 and vise-versa) before engaging in the execution, running the protocol, and finally flipping the result of the computation.
\end{proof}

\section{Motivating Examples}
\label{sec:motivatingExamples}

We begin with two examples manifesting benefits of choosing an analysis together 
with a differentially private protocol for computing it. 
In the first example, this paradigm yields more efficient protocols than the natural 
paradigm; in the second example, it yields simpler protocols.

\subsection{Binary Sum -- $\sqrt{n}$ Additive Error} \label{sec:MotivBinarySum}
We begin with a simple protocol for approximating $\SUM_n$ within $O(\sqrt{n}/\eps)$-additive approximation. This protocol is  well known as {\em Randomized Response}~\cite{W65}. 
We describe the protocol in the (non-interactive) local model, and it can be easily translated to a two round
(and $2n$ messages) $(n,\eps)$-differentially private distributed protocol by letting some arbitrarily designated party (say $p_1$) play the role of $C$.

Let $\Flip_{\alpha}(x)$ be a randomized bit flipping operator returning $x$
with probability $0.5+\alpha$ and $1-x$ otherwise, where $\alpha=\frac{\eps}{4+2\eps}$. The protocol proceeds as follows:  
\begin{enumerate}
\item Each party $p_i$ with private input $x_i \in \{0,1\}$ sends $z_i = \Flip_{\alpha}(x_i)$ to $C$.
\item
  \label{stp:C}
  $C$ locally computes and publishes $k=\sum_{i=1}^n z_i$.
\item
  \label{stp:each}
  Each party locally computes $\hat{f}= (k - (0.5-\alpha) n)/{2 \alpha}$.
\end{enumerate} 
A total of $O(n)$ messages and $O(n\log n)$ bits of communication are exchanged. To see that the protocol satisfies the privacy requirement of \defref{local-model-privacy-individual}, note that  
$$\frac{\Pr[\Flip_{\alpha}(1) = 1]}{\Pr[\Flip_{\alpha}(0)= 1]} = \frac{0.5+\alpha}{0.5-\alpha} = 1+\eps
\leq e^\eps,$$ and similarly $\Pr[\Flip_{\alpha}(0) =
0]/\Pr[\Flip_{\alpha}(1)= 0] \leq e^\eps$.
To see that the protocol approximates the sum function, note that 
$$\Exp[z_i]=\Exp[\Flip_{\alpha}(x_i)]=\left\{
 \begin{array}{ll}
 0.5+\alpha & \text{if\ } x_i =1 \\ 
 0.5-\alpha & \text{if\ } x_i =0.
 \end{array}
\right.$$
Thus,
$$\Exp[k] = (0.5+\alpha)\cdot\SUM({\myvec x}) + (0.5-\alpha)\cdot(n-\SUM({\myvec x}))
= 2\alpha\cdot\SUM({\myvec x}) + (0.5-\alpha)n,$$ 
and hence, 
$$\Exp[\hat{f}] = \Exp\left[\frac{k - (0.5-\alpha)n}{2\alpha}\right] =\SUM({\myvec x}).$$ 
By an application of the Chernoff bound, we get that $\hat{f}$ is an additive $(O(1),O(\sqrt{n}/\eps))$-approximation to $\SUM(\cdot)$, that is, with constant probability, the error
is $O(\sqrt{n}/\eps)$.

\begin{remark}
We next sketch an alternative $\eps$-differentially private protocol that $(O(1),\sqrt{n}/\eps)$-approximates $\SUM_n$:
\begin{enumerate}
\item Each party $p_i$ with private input $x_i \in \{0,1\}$ samples $y_i\sim
\Lap(1/\eps)$ and sends $z_i = x_i+y_i$ to $C$.
\item $C$ locally computes $\hat{f} =\sum_{i=1}^n z_i$ and publishes the
result. 
\end{enumerate}
The privacy of the protocol follows from the arguments in \secref{dp}.
\end{remark}

\begin{remark}
The above constructions result in {\em symmetric} approximations to $\SUM(\cdot)$ (i.e., the output distribution depends solely on $\SUM(\cdot)$ and not on the specific assignment).  While these differentially private protocols use $O(n)$ messages, it can be shown that for such symmetric functions that no efficient SFE protocols for such functions exist (see \secref{SymmetricFunctions} for more details).
\end{remark}

\subsection{Distance from a Long Subsequence of 0's}

Our second function measures how many bits in a sequence $\myvec {x}$ of
$n$ bits should be set to zero to get an
all-zero consecutive subsequence of length $n^{\alpha}$. In other words, the function should return the minimum weight over all 
substrings of $\myvec {x}$ of length
$n^{\alpha}$ bits: $$\DIST_\alpha({\myvec x}) = \min_i\left(\sum_{j=i}^{i+n^{\alpha}-1} x_{j}\right).$$ 
For $t \leq n/2$ we present a $(t,\eps,\delta)$-differentially private protocol%
\footnote{
$(\eps,\delta)$-differential privacy is a generalization,
defined in~\cite{DKMMN06}, of
$\eps$-differential privacy
where it is only required that
$\Pr[\hat{f}({\myvec x}) \in \V] \leq e^\eps\cdot \Pr[\hat{f}({\myvec x'}) \in \V]+\delta\;.$
}
approximating $\DIST_\alpha({\myvec x})$ with additive error $\tilde O(n^{\alpha/3}/\eps)$.

In our protocol, we treat the $n$-bit string $\myvec x$
(where $x_i$ is held by party $p_i$) as a sequence of $n^{1-\alpha/3}$
disjoint intervals, each $n^{\alpha/3}$ bit long. Let ${i_1}, \dots,
{i_{n^{1-\alpha/3}}}$ be the indices of the first bit in each interval,
and observe that $\min_{i_k} (\sum_{j={i_k}}^{i_k+n^{\alpha}-1} x_{j})$ is
an $n^{\alpha/3}$ additive approximation of $\DIST_\alpha$.
The protocol for computing an approximation $\hat{f}$ to $\DIST_\alpha$ is sketched below.

\begin{enumerate}
\item Every party $p_i$ generates a random variable $Y_i$ distributed according to the normal distribution $N(\mu = 0, \sigma^2 = 2R/n)$ where 
$R = \frac {2 \log {(\frac{2}{\delta})}}{\eps^2}$, and shares $x_i+Y_i$ between the parties
$p_1,\ldots,p_{t+1}$ using an additive $(t+1)$-out-of-$(t+1)$ secret sharing
scheme%
\footnote{Shared secrets are taken from a finite domain by rounding the numbers
$\log n$ digits after the point. 
This yields no breach in privacy and adds a small magnitude of error.}.

\item
Every party $p_i$, where $1 \leq i \leq t+1$,
sums, for every interval of length $n^{\alpha/3}$,
the shares it got from the parties in the interval and sends this sum to
$p_1$.
\item
For every interval of length $n^{\alpha/3}$, party
$p_1$ computes the sum of the $t+1$ sums it got for the interval.
By the additivity of the secret sharing scheme, this sum is equal to 
$$S_k = \sum_{j={i_k}}^{i_k+n^{\alpha/3}-1} (x_{j} + Y_j)
= \sum_{j={i_k}}^{i_k+n^{\alpha/3}-1} x_{j} + Z_k,$$
where $Z_k= \sum_{j={i_k}}^{i_k+n^{\alpha/3}-1} Y_{j}$ (notice that
$Z_k\sim N(\mu=0,\sigma^2=2R)$).
\item $p_1$ computes $\min_{k} \sum_{j=k}^{k+n^{2\alpha/3}} S_k$
and sends this output to all parties.
\end{enumerate}
Using the analysis of~\cite{DKMMN06}, this protocol is a
$(t,\eps,\delta)$-differentially private protocol when $2t
<n$. Furthermore, it can be shown that with high probability the additive
error is $\tilde O(n^{\alpha/3}/\eps)$.  To conclude, we showed a
simple $3$ round protocol for $\DIST_\alpha$.

This protocol demonstrates two advantages of the paradigm of choosing what
and how together. First, we choose an approximation of  $\DIST_\alpha$
(i.e., we compute the minimum of subsequences starting at a beginning of an
interval). This approximation reduces the communication in the
protocol. Second, we leak information beyond the output of the
protocol, as $p_1$ learns the sums $S_k$'s%
\footnote{One can use the techniques of~\cite{DFNT06} to avoid leaking
these sums while maintaining a constant number of rounds, however the
resulting protocol is less efficient.}.

\section{Lowerbounds on the Error of Binary Sum and Gap-Threshold in the Local Model}
\label{sec:local}

We prove that any $\ell$-round $\eps$-differentially private protocol in the local model for computing the binary sum function must exhibit an additive error of $\Omega(\sqrt{n}/\tilde{O}(\ell))$.
By \corref{reduction-gap-sum} and \clmref{reduction-gapK-gap0}, it suffices to prove that such a protocol can only compute $\GAP_{0,\tau}$  for $\tau = \Omega(\sqrt{n}/\tilde{O}(\ell))$ (i.e., the parameter $\kappa$ is set to zero). For that, we show that there are two input vectors -- one containing  $\Omega(\sqrt{n})$ ones, and the other is all zero -- for which the curator sees similar distributions on the messages, and hence must return similar answers. 

We will begin by having the non-zero vector be distributed according to a probability distribution $\A$ (on $n$-bit vectors). This implies that a specific choice for this vector exists. In the following we set
\begin{eqnarray}
\label{eq:alpha}
  \alpha & \eqdef & \frac{1}{\eps\sqrt{d n}},
\end{eqnarray}
where $d > 1$ (the value of $d$, which is a function of the number of rounds in the protocol $\ell$, is determined later).

\begin {notation} 
Define the distribution $\A$ on inputs from $\set{0,1}^n$ as follows: a vector $\myvec{x}=\vect{x_1,\dots,x_n}$ is chosen, where $x_i=1$ with probability $\alpha$ and $x_i=0$ with probability $(1-\alpha)$ (each input $x_i$ is chosen independently).

We use $\myvec X$ to identify the random variable representing the joint input and $X_i$ for the random variable corresponding to its $i$-th coordinate. The notation $\pr_\A[\cdot]$ is used when a probability over the choice of $\myvec{X}$ from $\A$ is considered.  For a set $D$ of possible curator's views we use the notation $\pr_\A[D]$ to denote the
probability of the event that the view of the curator falls in $D$ when the joint input $\myvec{X}$
is chosen according to $\A$.
\end {notation}

\paragraph{Main steps of the proof:} In \secref{NI}, we analyze properties of non-interactive differentially
private protocols in the local model, and show that a curator, trying to
distinguish between an input chosen according to distribution $\A$ and the
all zero input, fails with constant probability.  In \secref{I} we
generalize this analysis to interactive protocols in the local model.
In \secref{Local} we complete the proof of the lowerbound on the
gap-threshold function in the local model.

\subsection{Differentially Private Protocols in the Non-Interactive Local Model}
\label{sec:NI}

Consider protocols in the non-interactive local model where each party
holds an input $x_i\in \set{0,1}$ and independently applies an algorithm
$S_i$ (also called a sanitizer) before sending the sanitized result $c_i$
to the curator.  We want to prove that if each $S_i$ is
$2\eps$-differentially private for some $0< \eps \le 1$%
\footnote{We can relax the condition  $\eps \le 1$ by a condition $\eps \le \eps_0$ for any constant $\eps_0 \geq 1$. This would affect some of the constants in the calculations below.},
then the curator errs with constant probability  when trying to distinguish
between an input chosen according to distribution $\A$ and ${\myvec 0}$
(where ${\myvec 0}$ is the vector $0^n)$%
\footnote{We consider protocols that are $2\eps$-differentially private to simplify the notation in \secref{I}.}.

For every possible view $\myvec{c} = \vect{c_1,\dots, c_n}$ of the curator $C$, we consider the ratios of the probability of receiving messages according to $c$ when the input is chosen according to $\A$ and when it is ${\bf 0}$. The probability is over the randomness of the protocol, and over the choice according to distribution $\A$ where specified:
\begin{equation}
\label{eqn:r}
r(\myvec{c}) \eqdef \frac{\pr_\A\Big[\View_{C}({\myvec X})=\myvec{c}\ \Big]}%
                       {\pr\Big[\View_{C}({\myvec 0})=\myvec{c}\ \Big]}
  \quad \mbox{and} \quad
r_i(c_i) \eqdef \frac{\pr_\A\Big[S_i(X_i)={c_i}\Big]}{\pr\Big[S_i(0)=c_i\Big]}.
\end{equation}
Since in a non-interactive protocol $\pr[\View_{C}({\myvec 0})=\myvec{c}] = \prod_{i=1}^n \pr_\A[S_i(0)={c_i}]$ (the sanitizers $S_i$ use independent randomness) and 
$\pr_\A[\View_{C}({\myvec X})=\myvec{c}] = \prod_{i=1}^n \pr_\A[S_i(X_i)={c_i}]$ (the sanitizers $S_i$ use independent randomness and the entries of the random variable $\myvec X$ are chosen independently), we have that 
\begin{equation} \label{eq:rprod}
r(\myvec{c})=\prod_{i=1}^n r_i(c_i).
\end{equation}

We next show that if the inputs are selected according to $\A$, then with constant probability $r(\myvec{c})$ is bounded by a constant. In other words, for those views $c$ of the curator that are likely when inputs are selected according to $\A$, the probability of seeing $\myvec c$ when the protocol is executed with inputs selected according to $\A$ is similar to the probability of seeing $\myvec c$ when the protocol is executed with inputs set to zero.

Define a random variable $\myvec{C}=(C_1,\ldots,C_n)$ where $C_i = S_i(X_i)$ and $X_i$ is chosen 
according to the distribution $\A$. 
Defining the random variables $V_i\eqdef \ln r_i(C_i)$, we can write for every $\eta>0$:
\begin{equation}
\label{eqn:ln}
\pr_\A[r(\myvec{C}) > \eta]  = \Pr_\A\left[\prod_{i=1}^n r_i(C_i) > \eta\right] 
=\pr_\A\left[\sum_{i=1}^{n}{V_i} >\ln \eta\right],
\end{equation}
where the first equality is by \eqref{eq:rprod} above. 
 In the next two lemmas we show that each variable $V_i$ is bounded, and bound its expectation. Both proofs use the $2\eps$-differential privacy of the sanitizers. These bounds are then used with the Hoeffding bound in \lemref{Hoeffding_Ratio_AtoB} where we bound $\pr_\A[r(\myvec{C}) > \eta]$.

\begin{lemma}
\label{lem:bound}
 For every $i$ and for any $0 < \eps \le 1$, with probability one, $1-2\alpha\eps \leq r(c_i) \leq 1 + 
4\alpha\eps$ and $-4\alpha\eps \leq V_i \leq 4\alpha\eps$.
\end{lemma}
\begin{proof}
For every $i$ and every value $c_i$,
$$ r_i(c_i) = \frac{\pr_{\A}[S_i(X_i)=c_i]}{\pr[S_i(0)=c_i]} = \frac{\alpha\pr[S_i(1)=c_i]+(1-\alpha)\pr[S_i(0)=c_i]}{\pr[S_i(0)=c_i]} = 1+\alpha\left(\frac{\pr[S_i(1)=c_i]}{\pr[S_i(0)=c_i]}-1\right).$$
Using $e^{-2\eps}\leq\frac{\pr[S_i(1)=c_i]}{\pr[S_i(0)=c_i]}\leq e^{2\eps}$, we get that 
$$ 1+\alpha(e^{-2\eps}-1) \leq r_i(c_i) \leq 1+\alpha(e^{2\eps}-1).$$

Using $e^{2x} < 1+4x$ and $1-e^{-2x} < 2x$ for $0 < x \leq 1$, we get 
$1-2\alpha\eps \leq r_i(C_i) \leq 1+4\alpha\eps$. Recall that $V_i = \ln r_i(C_i)$. Using $\ln(1+x) \leq x$ and $\ln(1-x) \geq -2x$ for $0 \leq x \leq 0.5$ and noting that $\alpha = 1/(\eps\sqrt{dn})$ and hence $4\alpha\eps \ll 0.5$, we get that $-4\alpha\eps \leq V_i \leq 4\alpha\eps$.
\end{proof}

\begin{lemma}
\label{lem:expected}
For every $i$ and for any $0 < \eps \le 1$, $$\expected[V_i] \leq 32\alpha^2\eps^2.$$
\end{lemma}
\begin{proof}
For the proof, we assume that the output of $S_i$ is in a countable set.
Let $$B_b\eqdef\set{c_i:r_i(c_i)=1+b}\quad\quad \mbox{for}~-2\alpha\eps \leq b \leq 4\alpha\eps.$$
\lemref{bound} implies that these are the only values possible for $b$.
By the definition of $r_i$, for every $c_i\in B_b$,
$$\frac{\pr_\A[S_i(X_i)=c_i]}{\pr[S_i(0)=c_i]}=r(c_i)=1+b,$$
and hence,
\begin{eqnarray}
\pr[S_i(0)\in B_b]=\frac{\pr_\A[S_i(X_i)\in B_b]}{1+b}\leq (1-b+2b^2)\cdot\pr_\A[S_i(X_i)\in B_b]. \label{eq:split_b}
\end{eqnarray}
Let $\beta=2\alpha\eps$.
We next bound  $\expected[V_i]$:
\begin{eqnarray}
\expected[V_i] & = & \expected_\A[\ln r(C_i)]
 =  \sum_{-\beta \leq b \leq 2\beta}\pr_\A[S_i(X_i)\in B_b]\cdot\ln(1+b) \nonumber \\ 
& \leq & \sum_{-\beta \leq b \leq 2\beta}\pr_\A[S_i(X_i)\in B_b]\cdot b \label{eq:explain_ln} \\
& = &  \sum_{-\beta \leq b \leq 2\beta}\pr_\A[S_i(X_i)\in B_b] \cdot (1 + 2b^2)
      -  \sum_{-\beta \leq b \leq 2\beta}\pr_\A[S_i(X_i)\in B_b] \cdot (1-b+2b^2) \label{eq:replaced}
\end{eqnarray}
Where \eqref{eq:explain_ln} follows by $\ln(1+b)\leq b$.
Using \eqref{eq:split_b} we can replace the second term in \eqref{eq:replaced} by $\sum_{-\beta \leq b \leq 2\beta}\pr[S_i(0)\in B_b]$ and get
\begin{eqnarray*}
\expected[V_i] & \leq & (1+ 2(2\beta)^2) \sum_{-\beta \leq b \leq 2\beta}\pr_\A[S_i(X_i)\in B_b]  
- \sum_{-\beta \leq b \leq 2\beta}\pr[S_i(0)\in B_b]\\
& = & (1+ 8\beta^2) \cdot \pr_\A[S_i(X_i)\in \cup_b B_b] - \pr[S_i(0)\in \cup_b B_b] \\
&\leq & (1 + 8\beta^2)\cdot 1 - 1 = 8\beta^2=32\alpha^2\eps^2. 
\end{eqnarray*}
\end{proof}

\medskip

By \lemref{expected}, $\Exp [\sum_{i=1}^n V_i] = \sum_{i=1}^n{\Exp [V_i]} \leq 32\alpha^2
\eps^2 n=32/d$.  We next prove \lemref{Hoeffding_Ratio_AtoB}
which shows that  $\sum_{i=1}^n V_i$ is concentrated around this value. We use the Hoeffding bound:
\begin{theorem} [Hoeffding bound]
Let $V_1, \dots, V_n $ be independent random variables such that
$V_i \in [a, b]$ and $\sum_{i=1}^n \expected[V_i]=\mu$.
Then, for every $t > 0$,
$$\Pr\left[ \sum_{i=1}^n V_i - \mu \geq t\right] \leq  \expon{-\frac{2\,t^2}{n(b - a)^2} }.$$
\end{theorem}

\begin{lemma}
\label{lem:Hoeffding_Ratio_AtoB}
$\pr_\A[r(\myvec{C}) > \expon{\nu/ d}] <  \expon{-(\nu-32)^2/32d}$ for every $\nu > 32$.
\end{lemma}        
\begin{proof} 
By \eqnref{ln},  \lemref{bound}, \lemref{expected}, and substituting $\alpha = \frac{1}{\eps\sqrt{dn}}$:
\begin{eqnarray*}
\pr_\A[r(\myvec{C}) > \expon{\nu/ d}] & = &
\pr_\A\left[\sum_{i=1}^{n}{V_i} >\frac{\nu} {d}\right] \\
& = &
\pr_\A\left[\sum_{i=1}^{n}{V_i} -\sum_{i=1}^n\expected V_i > \frac{\nu} {d}-\sum_{i=1}^n\expected V_i\right]\\
& \leq &
\pr_\A\left[\sum_{i=1}^{n}{V_i} -\sum_{i=1}^n\expected V_i
  > \frac{\nu} {d}-n\cdot 32\alpha^2\eps^2\right] \\
& \leq & \expon{-\frac{2\,( \frac{\nu} {d}-n\cdot 32\alpha^2\eps^2)^2}%
          {64\,n\,\alpha^2\,\eps^2}}\\
&  = & \expon{-(\nu-32)^2 / 32 d}.
\end{eqnarray*}
\end{proof}

We now rephrase \lemref{Hoeffding_Ratio_AtoB} in a way that would be more convenient for our argument in the next section. 
Let $\Pi$ be a $2\eps$-private, non-interactive, local protocol, where $0<\eps \le 1$. 
For a possible curator's view $\myvec{c}$, let 
$$p_\A(\myvec{c}) = \pr_\A[\View_{C}(\myvec{X})=\myvec{c}] \quad \quad \mbox{and} \quad \quad  p_\myvec{0}(\myvec{c}) = \pr[\View_{C}(\myvec {0})=\myvec{c}],$$
where in $p_\A(\myvec{c})$ the probability is taken over the choice of $\myvec{X}$
according to the distribution $\A$ and the randomness of $\Pi$, and
in $p_\myvec{0}(\myvec{c})$ the probability is taken over the randomness of $\Pi$.
The following corollary follows from \lemref{Hoeffding_Ratio_AtoB} and the definition of $r$ in \eqnref{r}.

\begin{corollary}
\label{cor:non-interactive}
Assume we execute $\Pi$ with input sampled according to distribution $\A$, then for every $\nu > 32$, with probability at least $1-\expon{-(\nu-32)^2/32 d}$, the curator's view satisfies:
$$p_\A(\myvec{c}) \leq \expon{\nu/ d} \cdot p_\myvec{0}(\myvec{c}),$$
where the probability is taken over the random choice from $\A$ and the randomness of \/ $\Pi$.
\end{corollary}        

\subsection{Differentially Private Protocols in the Interactive Local Model}
\label{sec:I}

In this section we generalize \corref{non-interactive} to interactive local
protocols where each party holds an input $x_i\in \set{0,1}$. The structure of our argument is as follows:
\begin{enumerate}
\item We decompose an $\ell$-round $\eps$-differentially private protocol $\Pi$ into $\ell$ non-interactive, local protocols, and prove that each of the $\ell$ protocols is $2\eps$-differentially private. Thus, we can apply \corref{non-interactive} to each protocol.
\item We view the original protocol as a protocol between the curator and a single party, simulating the other $n$ parties. 
In this protocol the curator's goal is to determine whether inputs are all zero or they are sampled according to $\A$. 
We apply a composition lemma to show that the curator's success probability does not increase by too much as $\ell$ grows. Clearly, this is true also for the original protocol.
\end{enumerate}

\subsubsection{A Composition Lemma}

Consider an interactive protocol, where a (deterministic) curator $C$ sends adaptive queries to a single (randomized) party $p$ 
holding a private input $x\in\{0,1\}$ in a similar setup to that of the local model (except that we make no requirement 
for $\eps$-differential privacy). We assume that the party $p$ is stateless and  that in each round $1\leq j\leq \ell$, the protocol proceeds as follows:
\begin{enumerate}
\item In the first phase of round $j$, the curator $C$ sends $p$ a message $q_j$
  (this message is also called the query); this message is a function of the round number $j$ and the messages the curator got  from $p$ in the previous rounds.
\item In the second phase of round $j$, party $p$ chooses fresh random coins and based on these coins and the query $q_j$ it computes a message $\V_j$ and sends it to the curator. We consider the randomized function computing the message $\V_j$ as an algorithm $A_j$, that is, $\V_j=A_j(x)$. 
\end{enumerate}

\begin{definition}
We say that a possible outcome $\V_j$ is $\eps$-good for algorithm $A_j$ if \/ $\Pr[A_j(1)=\V] \leq e^\eps \Pr[A(0)=\V]$, 
where the probabilities are taken over the randomness of algorithm $A_j$. 
An algorithm $A_j$ is $(\eps,\delta)$-good if \/ $\Pr[\mbox{$A_j(1)$ is $\eps$-good for $A_j$}] \geq 1-\delta$, 
where the probability is taken over the randomness of $A_j$.
\end{definition}

Let $\Pi$ be a protocol, as defined above, in which for every $j$ and every transcript of messages $\V_1,\ldots,\V_{j-1}$,
sent by $p$ in rounds $1,\ldots,j-1$, the curator $C$ replies with a query $q_j$, such that the algorithm $A_j$ resulting from $q_j$ is  an $(\eps,\delta)$-good algorithm.
Define a randomized algorithm $\hat{A}$ that simulates the interaction between $p$ and $C$, i.e., given input $x\in\{0,1\}$ 
it outputs a transcript $(q_1,\V_1,q_2,\V_2,\ldots,q_\ell,\V_\ell)$ sampled according to $\Pi(x)$.

\begin{lemma}
\label{lem:compo}
$\hat{A}$ is $(\ell\eps,1-(1-\delta)^\ell)$-good.
\end{lemma}
\begin{proof}
Choose a random transcript $(q_1,\V_1,q_2,\V_2,\ldots,q_\ell,\V_\ell)$, and let
$A_1,A_2,\ldots,A_\ell$ be the algorithms defined by this transcript. By our assumptions all these algorithms are $(\eps,\delta)$-good.
Thus, with probability at least $(1-\delta)^\ell$, the
transcript $\hat{\V}=(q_1,\V_1,q_2,\V_2,\ldots,q_\ell,\V_\ell)$ is 
such that $\V_j$ is $\eps$-good for $A_j$ for all $1 \leq j\leq \ell$. It suffices, hence, to prove that when that happens the transcript $\hat{\V}$ is $\ell\eps$-good for $\hat{A}$, and indeed,
\begin{eqnarray*}
\Pr[\hat{A}(1) =  (q_1,\V_1,q_2,\V_2,\ldots,q_\ell,\V_\ell)] & = &\prod_{j=1}^\ell \Pr[A_j(1)=\V_j] \\
& \leq & \prod_{j=1}^\ell e^\eps \cdot \Pr[A_j(0)=\V_j] \\
& = & e^{\ell\eps} \cdot \prod_{j=1}^\ell \Pr[A_j(0)=\V_j] \\
& = & e^{\ell\eps}\cdot \Pr[\hat{A}(0) = (q_1,\V_1,q_2,\V_2,\ldots,q_\ell,\V_\ell)].
\end{eqnarray*}
The first and last equalities follow by independence and by the fact that the curator is deterministic. The inequality follows by the $\ell$-goodness of $\V_1,\ldots,\V_\ell$.
\end{proof}

\subsubsection{The Main Lemma}

Let $\Pi$ be an $\ell$-round, local, $\eps$-differentially private protocol, where $0<\eps \le 1$. 
For a possible curator's view $\myvec{c}$, let 
$$p_\A(\myvec{c}) = \pr_\A[\View_{C}(\myvec{X})=\myvec{c}] \quad \quad \mbox{and} \quad \quad  p_\myvec{0}(\myvec{c}) = \pr[\View_{C}(\myvec {0})=\myvec{c}],$$
where in $p_\A(\myvec{c})$ the probability is taken over the choice of $\myvec{X}$
according to the distribution $\A$ and the randomness of $\Pi$, and
in $p_\myvec{0}(\myvec{c})$ the probability is taken over the randomness of $\Pi$.

\begin{lemma}
\label{lem:interactive}
Assume we execute $\Pi$ with input sampled according to distribution $\A$, then for every $\nu > 32$, with probability at least $1-\ell\cdot \expon{-(\nu-32)^2/32 d}$, the curator's view satisfies:
$$p_\A(\myvec{c}) \leq  \expon{\ell \nu /d} \cdot p_\myvec{0}(\myvec{c}),$$
where the probability is taken over the random choice from distribution $\A$ and the randomness of \/ $\Pi$.
\end{lemma}
\begin{proof}
Recall that in the interactive local model, a protocol is composed of $\ell$-rounds where in each round the curator sends a query to each party and the
party sends an answer. We modify the protocol, to make the parties stateless, by introducing the following changes to the interaction between the curator and every party $p_i$. Both changes do not affect the privacy of the protocol, nor its outcome.
\begin{enumerate}
\item In round $j$ the curator sends all queries and answers $q_{1},a_{1},$ $\dots,a_{j-1},q_{j}$ it sent and received from $p_i$ 
in previous rounds%
\footnote{To simplify notation, we omit the subscript $i$ from the queries
and answers.}.
\item Party $p_i$ chooses a fresh random string in each round, that is, in round $j$, party $p_i$ chooses with uniform distribution a random string that is consistent with the queries and answers it got in the previous rounds (since we assume that the parties are computationally unbounded, such choice is possible). Party $p_i$ uses this random string to answer the $j$th query.  In other words, we can consider $p_i$ as applying an algorithm $A_{j}$ to compute the $j$th answer; this algorithm depends on the previous queries and answers and uses an independent random string $r_j$.
\end{enumerate}

We next claim that $A_{j}$ is $2\eps$-differentially private.  That is,
we claim that the probability that $a_{j}$ is generated given the previous
queries and answers is roughly the same when $p_i$ holds the bit $0$ and
when $p_i$ holds the bit $1$.  
 For a transcript $c$ of the first $j$ rounds between $p_i$ and the curator $C$ and for $x_i \in \{0,1\}$,
we denote by $R^{x_i}_{c}$ the set of all random strings $r$, such that 
$p_i$ with private input $x_i$ and random input $r$ 
sends at each round messages according to $c$, provided it 
received all messages according to $c$ in previous rounds.
Recall that $\pr[r_j \in R^{x_i}_{c}]$ is denoted $\alpha_i^{c}(x_i)$.
Let $c_{j} = q_{1},a_{1},\dots,q_{j-1},a_{j-1},q_j,a_j$ be a $j$-round transcript
and let $c'_{j} = q_{1},a_{1},\dots,q_{j-1},a_{j-1},q_j$ be the prefix
of $c_{j}$ without the $j$th round answer $a_j$ (that is, $c_{j} = c'_{j} \circ a_j$).
Note that, since $r_j$ must be consistent with the $c'_{j}$, it holds for every $x_i \in \set{0,1}$
that ${\Pr[A_{j}(x_1)= a_{j}]} = {\Pr[r_j \in R^{x_1}_{c_j} |r_j \in R^{x_1}_{c'_{j}}]}$.
We therefore need to show that 
\begin {equation*}
e^{-2\eps} \le   \frac{\Pr[A_{j}(1)= a_{j}]}{\Pr[A_{j}(0) = a_{j}]}  = 
\frac{\Pr[r_j \in R^{1}_{c_j} |r_j \in R^{1}_{c'_{j}}]}{\Pr[r_j \in R^{0}_{c_j} |r_j \in R^{0}_{c'_{j}}]}
\le e^{ 2\eps},
\end {equation*}
To show that, we use the following two
facts, which follow from \defref{local-model-privacy-individual}:

\begin {equation}\label{eqn:private-j}
e^{-\eps} \le  \frac{\alpha_i^{c_j}(1)}{\alpha_i^{c_j}(0)} = \frac{\Pr[r_j \in R^{1}_{c_j}]}{\Pr[r_j \in R^{0}_{c_j}]}
\le e^{ \eps},
\end {equation}
and 
\begin {equation} \label{eqn:private-j-1}
e^{-\eps} \le \frac{\alpha_i^{c'_j}(1)}{\alpha_i^{c'_j}(0)} = \frac{\Pr[r_j \in R^{1}_{c'_{j}}]}{\Pr[r_j \in R^{0}_{c'_{j}}]}
\le e^{\eps}
\end {equation}
Hence, we have
\begin{eqnarray*}
r \eqdef \frac{\Pr[A_{j}(1)= a_{j}]}{\Pr[A_{j}(0) = a_{j}]}
& =  & \frac {\Pr[r_j \in R^{1}_{c_j} \,\wedge\, r_j \in R^{1}_{c'_{j}}]} {\Pr[r_j \in R^{1}_{c'_{j}}]} \cdot 
    \frac {\Pr[r_j \in R^{0}_{c'_{j}}]}{\Pr[r_j \in R^{0}_{c_j} \,\wedge\, r_j \in R^{0}_{c'_{j}}]} \\
& =  & \frac {\Pr[r_j \in R^{1}_{c_j}]} {\Pr[r_j \in R^{1}_{c'_{j}}]} \cdot 
    \frac {\Pr[r_j \in R^{0}_{c'_{j}}]}{\Pr[r_j \in R^{0}_{c_j}]} = 
    \frac{\alpha_i^{c_j}(1)}{\alpha_i^{c_j}(0)}\cdot
                \frac{\alpha_i^{c'_j}(0)}{\alpha_i^{c'_j}(1)}.         
\end{eqnarray*}

Now, by using the right inequality in \eqnref{private-j} and the left inequality in \eqnref{private-j-1}, we get 
that $r \le e^{2 \eps}$ and similarly, by using the left inequality in \eqnref{private-j} and 
the right inequality in \eqnref{private-j-1}, we get that $r \geq e^{-2 \eps}$.
Thus, the answers of the $n$ parties in round $j$ are $2 \eps$-private,
and we can apply \corref{non-interactive} to the concatenation of the $n$ answers.

We now use the above protocol to construct a protocol $\Pi_1$
between a single party,  holding a one bit input $x$ and a
curator. Throughout the execution of the protocol the party simulates all
$n$ parties as specified by the original protocol $\Pi$ (i.e., sends messages to
the curator with the same distribution as the $n$ parties send them).  If
the bit of the party in $\Pi_1$ is $1$ it chooses the $n$ input bits of the
$n$ parties  in $\Pi$ according to distribution $\A$. If the bit of the
party in $\Pi_1$ is $0$ it chooses the $n$ input bits of the $n$ parties  in
$\Pi$ to be the all-zero vector. By \corref{non-interactive}
we can apply the composition lemma --
\lemref{compo} -- to the composition of the $\ell$, $2\eps$-differentially 
private, non-interactive protocols and the lemma follows.
\end{proof}

\begin{corollary}
\label{cor:prB}
Let $0 < \eps \leq 1$. For every $\nu>32$ and for every set $D$ of
views in an $\ell$-round, $\eps$-differentially private, local
protocol,
$$\pr[\View_{C}({\myvec 0})\in D] \geq
\frac{
\pr_\A[\View_{C}({\myvec X})\in D]- \ell\cdot\expon{-(\nu-32)^2 /32 d}}
{\expon{\ell\nu /d}}.
$$
\end{corollary}
\begin{proof}
Let 
$$D_1=\set{\myvec{c} \in D: \pr_\A[\View_{C}({\myvec X})=\myvec{c}]\leq \expon{\ell\nu /d}\pr[\View_{C}({\myvec 0})=\myvec{c}]}$$
and
$$D_2=\set{\myvec{c} \in D: \pr_\A[\View_{C}({\myvec X})=\myvec{c}] > \expon{\ell\nu /d}\pr[\View_{C}({\myvec 0})=\myvec{c}]}.$$
That is, $D_2 = D \setminus D_1$. By \lemref{interactive},
$\pr_\A[\View_{C}({\myvec X}) \in D_2]\leq  \ell \expon{-(\nu-32)^2 /32 d}$, and, furthermore,
$\pr[\View_{C}({\myvec 0}) \in D_1]\geq \pr_\A[\View_{C}({\myvec X}) \in D_1]/\expon{\ell\nu /d}$.
Thus,
\begin{eqnarray*}
\pr[\View_{C}({\myvec 0}) \in D]\geq \pr[\View_{C}({\myvec 0}) \in D_1]  & \geq & 
\frac{\pr_\A[\View_{C}({\myvec X}) \in D_1]}{e^{\ell\nu /d}} \\ & = &
\frac{\pr_\A[\View_{C}({\myvec X}) \in D]-\pr_\A[\View_{C}({\myvec X}) \in D_2]}{e^{\ell\nu /d}} \\
& \geq &
\frac{\pr_\A[\View_{C}({\myvec X}) \in D]-\ell e^{-(\nu-32)^2 /32 d}}{e^{\ell\nu /d}}. \quad
\end{eqnarray*}
\vspace*{-0.5cm}
\end{proof}

\subsection{Completing the Lowerbound for Gap-Threshold and Sum in the Local Model}
\label{sec:Local}

We now complete the proof that in any $\ell$-round, 
$\eps$-differentially private, local protocols for the gap-threshold function,
namely, $\GAP_{0,\tau}$, if $\tau \ll \sqrt{n}$ and $\ell$ is small, then 
the curator errs with constant probability.

Recall that we constructed the distribution $\A$ in which each bit in the input is chosen (independently at random) to be one with probability $\alpha$ and zero with probability $1-\alpha$.  
\lemref{Chernoff_XbyA}, which follows from a standard Chernoff bound argument, states that when generating a vector $\vect{X_1,\dots,X_n}$ according to $\A$, the sum $\sum_{i=1}^n X_i$ is concentrated around its expected value, which is $\alpha n$ (recall that $\alpha = 1/(\eps\sqrt{d n})$).
We apply the following Chernoff bound: Given $n$ zero-one random variables $X_1, \dots, X_n$ and $0 < t <1$, $\pr\left[\sum_{i=1}^{n}{X_i} \le (1-t)\mu \right] < \expon{-\frac{t^2 \mu}{2}}$, where $\mu =\sum_{i=1}^{n}{\expected[X_i]}$.
\begin{lemma}
\label{lem:Chernoff_XbyA}
$\pr_{\A}\left[\sum_{i=1}^{n}{X_i}\leq (1-\gamma) \alpha n\right]
< \expon{-\frac{\gamma^2\sqrt{n}}{2\eps\sqrt{d}}}$
for every $0 < \gamma <1$.
\end{lemma}
\begin {proof}
We use the above bound with $\mu = \alpha n= \frac{\sqrt{n}}{\eps\sqrt{d}}$.
Thus,
\begin{eqnarray*}
\pr_{\A}\left[\sum_{i=1}^{n}{X_i}\leq(1-\gamma)\alpha n \right] & < & \expon{-\frac{\alpha n \gamma^2}{2}} \\
& < & \expon{-\frac{\gamma^2 \sqrt{n}}{2\eps\sqrt{d}}}.
\end{eqnarray*}
\end {proof} 

On one hand, by \corref{prB}, the distributions on the outputs when the input vector is taken
from $\A$ and when it is the all zero vector are {\em not } far apart.
On the other hand, by \lemref{Chernoff_XbyA}, with high
probability the number of ones in the inputs distributed according to $\A$
is fairly big. These facts are
used in \thmref{local-gap} to prove the lowerbound.

\begin{theorem}
\label{thm:local-gap}
Let $0 < \eps \leq 1$. There exist constants $c>0$ and $p>0$ such that in any $\ell$-round, $\eps$-differentially private, 
local protocol for computing $\GAP_{0,\tau}$ for $\tau = c\frac{\sqrt{n}}{\eps \ell \sqrt{\log
\ell}}$ the curator errs with probability at least $p$.
\end{theorem}

\begin{proof}
Fix any $\ell$-round, $\eps$-differentially private, local protocol
for computing $\GAP_{0,\tau}$.
Recall that in the local model the curator is assumed to be 
deterministic.  Thus, the curator, having seen its overall view
of the execution of the protocol $c$, applies
a deterministic algorithm $G$ to $c$, where $G(c)$ is the output of the
protocol (which supposed to answer $\GAP_{0,\tau}(x_1,\dots,x_n)$ correctly).
Let $\tau=\alpha n/2 =\sqrt{n}/(2\eps\sqrt{d})$.

Denote by $D$ the set vectors of communication for which the curator
answers 1, i.e., $D  \eqdef \set{ \myvec{c} : G(\myvec{c}) = 1 }$.
The idea of the proof is as follows.
If the
probability of $D$ under the distribution $\A$ is small, then the curator
has a big error when the inputs are  distributed according to
$\A$. Otherwise, by \corref{prB}, the probability of $D$ when the 
inputs are all zero is big, hence the curator has a big error when the inputs
are the all-zero string. Formally, there are two cases:
     
\begin{description}
\item[Case I: \protect{$\pr_\A[D] < 0.99$.}] We consider the event that the sum of the inputs is at least
$\tau=\alpha n/2$ and the curator returns zero as an answer, that is, the curator errs.

We show that when the inputs are distributed according to $\A$ the probability of the complementary of this event is bounded away from 1.  By
the union bound the probability of the complementary event is at most $\pr_\A\left[\sum_{i=1}^{n}{X_i} < 0.5 \alpha n\right] +
\pr_\A[D] $.  By \lemref{Chernoff_XbyA},
\begin{eqnarray*}
\pr_\A[D]
   + \pr_\A\left[\sum_{i=1}^{n}{X_i} < 0.5 \alpha n\right]
 \le  0.99 + \expon{-0.25\sqrt{n}/(2\eps\sqrt{d})}
 \approx  0.99.
\end{eqnarray*}
Thus, in this case, with probability $\approx 0.01$ the curator errs.
                        
\item[Case II: \protect{$\pr_\A[D] \ge 0.99$.}] Here, we consider the event that the input is the all-zero string and the curator answers 1, that is,
the curator errs. 

We use \corref{prB} and show that when the inputs are all zero, the probability of this 
event is bounded away from 0 when taking $\nu=\theta(\ell\log \ell)$ and
$d=\ell\nu=\theta(\ell^2\log \ell)$,
\begin{eqnarray*}
\pr[\View_{C}({\myvec 0})\in D]
&\ge& \frac{\pr_\A[D] - \ell\expon{-(\nu-32)^2 / 32 d}}{\expon{\ell\nu /d}} 
> \frac{0.99-0.5}{\expon{1}}> 0.01.
\end{eqnarray*}
Thus, in this case, with probability 
at least $0.01$,
the curator errs. As $d=\theta(\ell^2\log \ell)$, we get that
$\tau=\sqrt{n}/(2\eps \sqrt{d})=
\theta(\sqrt{n}/(\eps \ell \sqrt{\log \ell}))$.
\end{description}
\end{proof}

By applying the local model variant of \corref{reduction-gap-sum}, we get our lowerbound for $\SUM_n$ as a corollary of \thmref{local-gap}:

\begin{corollary}
\label{cor:local-sum}
Let $0 < \eps \leq 1$. There exist constants $\delta>0$ and $p>0$ such that in any $\ell$-round, 
$\eps$-differentially private, local protocol for
computing $\SUM_n$ the curator errs with probability at least $p$ by 
at least $\frac {\delta \sqrt{n}}{\eps \ell \sqrt{\log \ell}}$.
\end{corollary}

\begin {proof}
Let $\Pi$ be an $\ell$-round, $\eps$-differentially private, local protocol for
computing $\SUM_n$, for which the curator errs by at most $\tau$ with probability at most $p$. 
By \corref{reduction-gap-sum} there exists an $\ell$-round, 
$\eps$-differentially private, local protocol for computing $\GAP_{0,2\tau}$ errs with probability at most $p$.
\end {proof}
                
\section{Lowerbounds for Binary Sum and Gap-Threshold in the Distributed Model}
\label{sec:dist}

We prove that, in any $\ell$-round, fixed-communication,
$(t,\eps)$-differentially private protocol computing the binary sum
with additive error less than $\sqrt{n}/\tilde{O}(\ell)$, the number of
messages sent in the protocol is $\Omega(nt)$. 
In the heart of our proof 
is the more general observation that in the information theoretic setting,
a party that has at most $t$ neighbors must protect its privacy with respect 
to his neighbors, since this set separates it from the rest of the parties.
Thus, any such party, is essentially as limited as any party participating in a 
protocol in the local communication model.

\subsection{From Distributed to Local Protocols}
\label{sec:reduction}

We start with the transformation of a distributed protocol, using a small number of messages
to a protocol in the local model.

\begin{lemma}
\label{lem:transformation}
If there exists an  $\ell$-round, fixed communication, $(t,\eps)$-differentially private
protocol that $(\gamma,\tau)$-approximates $\SUM_n$
sending at most $n(t+1)/4$ messages, then there exists an
$(\ell+1)$-round, $\eps$-differentially private protocol in the local model that 
$(\gamma,\tau)$-approximates $\SUM_{n/2}$.
\end{lemma}

\begin{proof}
The intuition behind the proof is that in the information theoretic model
if an input of a party affects the output, then the neighbors of this party
must learn information  on its input.  Recall that a party in a protocol $\Pi$ is
lonely if it communicates with at most $t$ other parties and it is popular
otherwise. If a party $p_i$ is lonely then it has most $t$ neighbors,
thus, from the privacy requirement in $(t,\eps)$-differentially private
protocols, they  are not allowed to learn ``too much" information on the
input of $p_i$. Therefore, the inputs of lonely parties cannot affect the
output of the protocol by too much, thus, since there are many lonely
parties, the protocol must have a large error.

Formally, assume that there is a distributed protocol $\Pi$ satisfying the conditions
in the lemma.  As the protocol sends at most $n(t+1)/4$ messages, 
the protocol uses at most $n(t+1)/4$
channels. Since each channel connects two parties, there are at least $n/2$
lonely parties.  We will construct a protocol in the local model which
$(\gamma,\tau)$-approximates $\SUM_{n/2}$ in two steps: 
In the first step, which is the main part of the proof, we
construct a protocol $\P$ in the local model which
                $(\gamma,\tau)$-approximates $\SUM_{n}$
                and only protects the privacy of the lonely parties.
In the second step, we fix the inputs of the popular parties and obtain a
    protocol $\P'$ for $n/2$ parties that protects the privacy of
    all parties.

\paragraph{First Step.}
We convert the distributed protocol $\Pi$ to a protocol $\P$ in the local
model as follows: Recall that in the local model each round consists of two phases 
where in the first phase the curator sends queries to the parties and in the second phase
parties send the appropriate responses. 
We hence have a single round in $\P$ for every round of $\Pi$ such that 
every message $m$ that Party $p_j$ sends to Party $p_k$ in round $i$ in
protocol $\Pi$, Party $p_j$ sends $m$ to the curator in round $i$ and
the curator sends $m$ to Party $p_k$ in the first phase of round $i+1$. Finally, at the end
of the protocol Party $p_1$ sends the output to the curator.

We next prove that $\P$ protects the privacy of lonely parties.
Without loss of generality, let $p_1$ be a lonely party, let $T$ be the set of size
at most $t$ containing the neighbors of $p_1$, and let
$R=\set{p_1,\dots,p_n}\setminus (T\cup\set{p_1})$.
See \figref{p1} for a description of these sets.
Fix any neighboring vectors
of inputs $\myvec{x}$ and $\myvec{x'}$ which differ on $x_1$.  The view of
the curator in $\P$ contains all  messages sent in the protocol.  It
suffices to prove that for every view $v$,
\begin{eqnarray}
\label{eqn:views}
\Pr[\View^{\P}_\C(\myvec{x}) =v ]
       & \leq &e^\eps\cdot \Pr[\View^{\P}_\C(\myvec{x'}) = v ] \;
\end{eqnarray}
(by simple summation it will follow for every set of views $\V$).

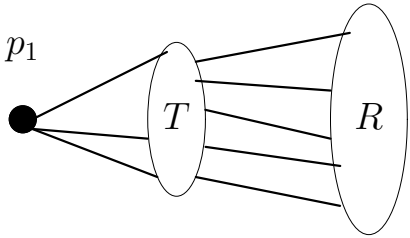
\begin{figure}[htbp]
  \begin{center}
    \setlength{\unitlength}{0.00066667in}
\begingroup\makeatletter\ifx\SetFigFont\undefined%
\gdef\SetFigFont#1#2#3#4#5{%
  \reset@font\fontsize{#1}{#2pt}%
  \fontfamily{#3}\fontseries{#4}\fontshape{#5}%
  \selectfont}%
\fi\endgroup%
{\renewcommand{\dashlinestretch}{30}
\begin{picture}(3355,1827)(0,-10)
\put(3047,831){\makebox(0,0)[b]{{\SetFigFont{14}{16.8}{\rmdefault}{\mddefault}{\updefault}$R$}}}
\put(1547,906){\ellipse{450}{1200}}
\put(3047,906){\ellipse{600}{1800}}
\thicklines
\path(422,906)(1472,1431)
\path(422,831)(1322,756)
\path(422,831)(1397,456)
\path(1697,1356)(2897,1581)
\path(1697,1206)(2747,1131)
\path(1773,692)(2823,542)
\path(1697,456)(2822,231)
\path(1772,981)(2747,756)
\put(347,1356){\makebox(0,0)[b]{{\SetFigFont{14}{16.8}{\rmdefault}{\mddefault}{\updefault}$p_1$}}}
\put(1547,831){\makebox(0,0)[b]{{\SetFigFont{14}{16.8}{\rmdefault}{\mddefault}{\updefault}$T$}}}
\thinlines
\put(347,906){\blacken\ellipse{212}{212}}
\put(347,906){\ellipse{212}{212}}
\end{picture}
}
    \caption{The partition of the parties to sets.}
    \label{fig:p1}
  \end{center}
\end{figure}

Fix a view $v$ of the curator.
For a set $A$, define $\alpha_A$ and ${\alpha'}_A$ as the probabilities in $\Pi$
that in each round the set $A$
with inputs from $\myvec{x}$ and $\myvec{x'}$ respectively
sends messages according to $v$ if it gets
messages according to $v$ in previous rounds
(these probabilities are taken over the random inputs of the parties in $A$).
Observe that if $p_1\notin A$, then $\alpha_A=\alpha'_A$ 
(since $\myvec{x}$ and $\myvec{x'}$ only differ on $x_1$).
By simulating $p_1$, $T$, $R$ by three parties
and applying \lemref{mult}, and by the construction of $\P$ from $\Pi$
\begin{eqnarray*}
\Pr\left[\View^{\P}_\C(\myvec{x})=v \right] & = & \alpha_{\set{p_1}}\cdot\alpha_T \cdot \alpha_{R}, \quad \mbox{and} \\
\Pr\left[\View^{\P}_\C(\myvec{x'})=v \right] & = & \alpha'_{\set{p_1}}\cdot\alpha'_T \cdot\alpha'_{R} = \alpha'_{\set{p_1}}\cdot\alpha_T \cdot
\alpha_{R}.
\end{eqnarray*}
Thus, proving \eqnref{views} is equivalent to proving that
\begin{equation}
\label{eqn:alpha'}
\alpha_{\set{p_1}} \leq e^\eps \alpha'_{\set{p_1}}.
\end{equation}

We use the $(t,\eps)$-privacy of protocol $\Pi$ to prove \eqnref{alpha'}.
Let $v_T$ be the messages sent and received by the parties in $T$ in $v$.
As $T$ separates $p_1$ from $R$, the messages in
$v_T$ are all messages in $v$ except for the messages exchanged between parties
in $R$.
The view of $T$ includes the inputs of $T$ in $\myvec{x}$, the messages $v_T$,
and the random inputs $\myvec{r_T}=\set{r_i:p_i \in T}$.
For a set $A$, define $\beta_A$ and $\beta'_A$ as the probabilities that in $\Pi$
in each round the set $A$ with inputs from $\myvec{x}$ and $\myvec{x'}$ respectively
sends messages according to $v_T$ if it gets
messages according to $v_T$ in previous rounds.
Note that $\beta_{\set{p_1}}=\alpha_{\set{p_1}}$
and $\beta'_{\set{p_1}}=\alpha'_{\set{p_1}}$
by the definition of $\P$.
By simulating $p_1$, $T$, $R$ by three parties, where the random inputs of
$T$ are fixed to $\myvec{r_T}$,
and by \lemref{mult},
\begin{eqnarray*}
& \Pr[\View^\Pi_T(\myvec{x}) =(\myvec{x_T},\myvec{r_T},v_T) ] = & \alpha_{\set{p_1}}\cdot \beta_{R}, \quad \mbox{and} \\
& \Pr[\View^\Pi_T(\myvec{x'}) =(\myvec{x_T},\myvec{r_T},v_T) ] = & \beta'_{\set{p_1}}\cdot \beta'_{R} = \alpha'_{\set{p_1}}\cdot \beta_{R}.
\end{eqnarray*}
(recalling that $\myvec{x_T}=\myvec{x'_T}$).
The above probabilities are taken over the random strings of $R$ and $p_1$ when the random strings of $T$ are fixed to $\myvec{r_T}$.
The $(t,\eps)$-differential privacy of $\Pi$ implies that 
\begin{eqnarray*}
 \Pr[\View^\Pi_T(\myvec{x}) =(\myvec{x_T},\myvec{r_T},v_T) ] \le e^{\eps}\Pr[\View^\Pi_T(\myvec{x'}) =(\myvec{x_T},\myvec{r_T},v_T) ].
\end{eqnarray*}
Thus, $\alpha_{\set{p_1}} \leq e^\eps \alpha'_{\set{p_1}}$ and, therefore, $\P$
is $\eps$-differentially private with respect to inputs of lonely parties.
 
\paragraph{Second Step.}
There are at least $n/2$ lonely parties in $\Pi$; without loss of generality, parties $p_1,\dots,p_{n/2}$ are lonely.
We construct a protocol $\P'$ that
$(\gamma,\tau)$-approximates $\SUM_{n/2}$
by executing Protocol $\P$ where 
(i) Party $p_i$, where $1\leq i\leq n/2$, with input $x_i$ sends
messages in $\P'$ as the party $p_i$ with input $x_i$  sends them in $\P$; and
(ii)
In addition, the party $p_1$ in $\P'$ simulates all other $n/2$ parties in $\P$,
that is, for every $n/2 < i \leq n$, it chooses a random input $r_i$ for $p_i$
and in every round it sends to the curator
the same messages as $p_i$ would send with $x_i=0$ and $r_i$.
Since the curator sees the same view in $\P$ and $\P'$ and since the privacy of
lonely parties is protected in $\P$, the privacy of each of the $n/2$
parties in $\P'$ is protected.
Protocol $\P'$, therefore,  
$(\gamma,\tau)$-approximates $\SUM_{n/2}$
(since we fixed $x_i=0$ for $n/2 < i \leq n$ and $\P'$ returns the same output distribution of $\Pi$, which 
$(\gamma,\tau)$-approximates $\SUM_{n}$ for all possible inputs).
\end{proof}

\medskip

We are now ready to state the main theorem of this section.

\begin{theorem}
\label{thm:dist}
Let $0 < \eps \leq 1$. 
There exist constants $\delta>0$ and $\gamma>0$ such that in any $\ell$-round, fixed-communication, 
$(t,\eps)$-differentially private protocol for approximating $\SUM_n$  
that sends at most $n(t+1)/4$ messages the protocol errs 
with probability at least $\gamma$ by 
at least $\frac {\delta \sqrt{n}}{\eps \ell \sqrt{\log \ell}}$.
\end{theorem}

\begin {proof}
Assume, for sake of contradiction, that there is an $\ell$-round,
$(t,\eps)$-differentially private protocol $\Pi$ for computing $\SUM_n$,
which sends at most $n(t+1)/4$ messages and errs by at most  $\tau = \frac
{\delta \sqrt{n}}{\eps \ell \sqrt{\log \ell}}$ with  probability at
least $1-\gamma$.  By \lemref{transformation} there exists an
$(\ell+1)$-round,  $\eps$-differentially private, local protocol $\P$ for
computing $\SUM_{n/2}$ which errs by  at most $\tau = \frac {\delta
\sqrt{n}}{\eps \ell \sqrt{\log \ell}} =  \frac {\sqrt{2}\delta
\sqrt{n/2}}{\eps \ell \sqrt{\log \ell}}$ with  probability at least
$1-\gamma$. This contradicts \corref{local-sum}.
\end {proof}

\medskip

\thmref{dist} suggests that whenever we require that the error of
a differentially private protocol for approximating $\SUM$
to be of magnitude smaller than $\sqrt{n}/\eps$, there is no reason to
relinquish the simplicity of the natural paradigm for constructing
protocols.  In this case, it is possible to construct relatively simple
efficient SFE protocols, which use $O(nt)$ messages, and compute an
additive $(O(1/\eps),O(1))$-approximation of $\SUM$.

\begin {remark}
It can also be shown that in any $\ell$-round, fixed-communication,
$(t,\eps)$-differentially private protocol computing the $\GAP_{\kappa,\tau}$, 
for any $0 \le \kappa \le n-\tau$, the number of messages sent in the protocol is $\Omega(nt)$,
 for $\tau = \sqrt{n}/\tilde{O}(\ell)$.
To show this, use the ideas similar to those of \lemref{transformation} and apply 
\thmref{local-gap} to assert that any $\ell$-round, fixed-communication,
$(t,\eps)$-differentially private protocol computing the $\GAP_{0,\tau}$, 
the number of messages sent in the protocol is $\Omega(nt)$, for $\tau = \sqrt{n}/\tilde{O}(\ell)$.
Then, using \clmref{reduction-gapK-gap0}, infer that the same is true for general $\kappa$.
\end {remark}

\section{SFE for Symmetric Approximations of Binary-Sum} 
\label{sec:SymmetricFunctions}

In this section we show the advantage of using the alternative paradigm for constructing 
distributed differentially private protocols whenever we allow an $O(\sqrt{n}/\eps)$ approximation.
Recall that it is possible to construct differentially
private protocols for such approximations that use $2n$ messages and are secure against coalitions of size 
up to $t=n-1$ (see \secref{MotivBinarySum}). We next prove, using ideas from Chor and Kushilevitz~\cite{CK93}, that any SFE protocol 
for computing a symmetric approximation for $\SUM_n$,
using less than $nt/4$ messages, has error magnitude $\Omega (n)$. 


We first give the definition of SFE protocols computing a given randomized function $\hat{f}(\cdot)$. Here, again,
we only consider protocols  where all parties are honest-but-curious and compute the {\em same} output. 
The definition is given in the information-theoretic model. 

\begin {definition} [SFE] \label{def:SFE-honest-but-curious} 
Let $\hat{f} : (\{0,1\}^*)^n \to \{0,1\}^*$ be an $n$-ary 
randomized function. Let $\Pi$ be an $n$-party protocol for computing $\hat{f}$. For a coalition $T \su \{1, \ldots,n \}$,
the view of $T$ during an execution of $\Pi$ on $\myvec {x} = (x_1 \ldots x_n)$, denoted 
$\View_{T}({\myvec x})$, is defined as in \defref{dist-dp}, i.e., $\View_T(x_1,\ldots,x_n)$ is the random variable 
containing the inputs of the parties in $T$ (that is, $\set{x_i}_{i\in T}$), the random inputs of the parties in $T$, and 
the messages that the parties in $T$ received during the execution of the protocol with inputs 
${\myvec x} = (x_1,\ldots,x_n)$. 

We say that $\Pi$ is a $t$-secure protocol for $\hat{f}$ if there exists a randomized function, denoted $S$, 
such that for every $t' \le t$, for every coalition $T = \{i_1, \ldots,i_{t'} \}$, and for every inputs 
$\myvec {x} = (x_1 \ldots x_n)$, the following two random variables are identically distributed:

\begin{itemize}
\item $\set{S\vect{T,\vect{x_{i_1}, \ldots,x_{i_{t'}}}, o}, o}$, where $o$ is obtained first by sampling $\hat{f}(\myvec {x})$ (recall that $\hat{f}$ is a randomized function) and then $S$ is applied to $\vect{T,\vect{x_{i_1}, \ldots,x_{i_{t'}}}, o}$.
                        
\item $\set{\View_{T}({\myvec x}), \Output^{\Pi}(\View_{T}({\myvec x}))}$, where $\Output^{\Pi}(\View_{T}({\myvec x}))$ denotes the output during the execution represented in $\View_{T}({\myvec x})$.
\end{itemize}
\end{definition}

\begin{claim}
\label{clm:v_T}
Let $\vecy$ and $\vecz$ be two inputs and $T$ be a coalition of size at most $t$ such that $\hat{f}({\vecy})$ and $\hat{f}({\vecz})$ are identically distributed
and $y_i=z_i$ for every $i \in T$.
In every $t$-secure protocol for $\hat{f}$,
for any possible view $v_T$ of the set $T$,
it holds that $\pr\left[\View_{T}({\myvec y}) = v_T\right] = \pr\left[\View_{T}({\myvec z}) = v_T\right]$.
\end{claim}
\begin{proof}
Let $T=\set{i_1,\dots,i_{t'}}$ for $t'\leq t$.
The two random variables 
$\set{S \vect{T,\left(y_{i_1}, \ldots,y_{i_{t'}}\right), o}, o}$
and $\set{S \vect{T,\vect{z_{i_1}, \ldots,z_{i_{t'}}}, o}, o}$ 
(as defined in \defref{SFE-honest-but-curious}) are identically distributed
since $\hat{f}(\myvec y)$ and $\hat{f}(\myvec z)$ are identically distributed.
Hence, by the $t$-security of the protocol, so do 
$\set{\View_{T}({\myvec y}), \Output^{\Pi}(\myvec {y})}$ and $\set{\View_{T}({\myvec z}), \Output^{\Pi}(\myvec {z})}$.
\end{proof}

\begin{definition}[Symmetric Randomized Function]\label{def:symmetricF} 
We say that a randomized function $\hat{f}$ over domain $D$ with range $R$ is {\em symmetric}
if it does not depend on the ordering on the coordinates of the input, i.e., for every 
$(x_1, \dots, x_n) \in D^n$ and every permutation $\pi : [n] \rightarrow [n]$ 
the distributions (over $R$) implied by $\hat{f}(x_1, \dots, x_n)$ and by $\hat{f}(x_{\pi(1)}, \dots, x_{\pi(n)})$
are identical.
\end{definition}

Note that allowing $O(nt)$ messages, it is fairly straightforward to construct a symmetric $(t,\eps)$-differentially private protocol with constant ($O(1/\epsilon)$) additive error for $\SUM_n$, using the natural paradigm with, say, the $\eps$-private approximation described in \exampleref{binary-sum}. The following lemma shows that $\Omega(nt)$ messages are essential whenever a symmetric approximation for $\SUM_n$ is computed by an SFE protocol, even if it is not required to preserve differential privacy.

\begin{lemma}
Let $\hat{f}$ be a symmetric randomized function approximating $\SUM_n$ such that for every input vector
$\myvec x$, it holds that
$\Pr\left[\left | \hat{f}(\myvec x) - \SUM(\myvec x) \right | < n/4\right] < 1/2$, 
and let $t \le n-2$. Every fixed-communication $t$-secure protocol $\Pi$ 
for computing $\hat{f}$ uses at least ${n(t+1)}/{4}$ messages\footnote{We note 
that the lemma does not hold for non-symmetric functions. For example, 
we can modify the bit flip protocol described in \secref{motivatingExamples}
to an SFE protocol for a non-symmetric function, 
retaining the number of messages sent (but not their length): in \stpref{C} $p_1$ also sends 
${\myvec z} = (z_1,\ldots,z_n)$, and in \stpref{each} each $p_i$ locally outputs $\hat{f} + {\myvec z} 2^{-n}$, 
treating ${\myvec z}$ as an $n$-bit binary number.}.
\end {lemma}

\begin{proof}
Let $\Pi$ be a $t$-secure protocol computing $\hat{f}$ using less than ${n(t+1)}/{4}$ messages. 
Then, there are at least $n/2$ lonely parties in $\Pi$. The intuition for the proof is that a 
lonely party does not affect the computation, since its neighboring set, being smaller than $t+1$, would otherwise be able to infer information about its input. The proof is given in two steps. In the first step, we show that for any given lonely party $p_i$, for any fixed inputs for all other parties, and for any transcript $c$ of the protocol, the probability of $c$ being the transcript of the protocol when $x_i=0$ is exactly the same as the probability of $c$ being the transcript of the protocol when $x_i=1$. In the second step of the proof, we use this to show that with probability at least $1/2$, the protocol errs by $n/4$.

Without loss of generality, assume $p_1$ is lonely and 
assume $p_2$ is not a neighbor of $p_1$. Let $T$ be the set of $p_1$'s neighbors and let
$R=\set{p_1,\dots,p_n}\setminus (T\cup\set{p_1})$ (in particular, $p_2 \in R$). 
Recall that for a transcript $c$ we denote by $\alpha_1^c(x_1)$, the probability that $p_1$ is consistent with 
$c$ with input $x_1$, namely, the probability that $p_1$ with input $x_1$ sends at each
round messages according to $c$, provided it received all messages according to $c$ in previous rounds.
Our goal in the first part of the proof is to prove that for any transcript of the protocol $c$, it holds that 
$\alpha_1^c(0) = \alpha_1^c(1)$.
Toward this end, we pursue the following proof structure.
\begin{itemize}
\item
We first consider
two inputs $\myvec z$ and $\myvec y$ such that $\SUM(\myvec z) =
\SUM(\myvec y)$, $y_i=z_i$ for every $i \in T$,
but $y_1 = 0$ while $z_1 = 1$. For every communication $c$ exchanged in $\Pi$, denote $c_T$ to be the messages sent
and received by the parties in $T$.
By \clmref{v_T}, since $\hat{f}$ is symmetric, the probability of $c_T$ is the
same with $\myvec z$ and with $\myvec y$.
\item
We simulate the protocol
$\Pi$ by a three-party protocol $\Pi'$, where the parties are $p_1$, $T$,
and $R$, and each one of them simulates the respective set of parties in
$\Pi$.  We then use \lemref{mult} to
write the probability that $c_T$ is the communication exchanged in $\Pi$ as a
product of $\alpha_1^{c_T}(x_1),
\alpha_T^{c_T}(\myvec{x_T})$, and  $\alpha_R^{c_T}(\myvec {x_R})$, where
$\myvec {x_T}$ (respectively, $\myvec {x_R}$) are the inputs of parties in
$T$ (respectively, in $R$).
We conclude that
$$\alpha_1^{c_T}(y_1)\cdot\alpha_T^{c_T}(\myvec{y_T})\cdot\alpha_R^{c_T}({\myvec y_R}) = \alpha_1^{c_T}(z_1)\cdot\alpha_T^{c_T}(\myvec{z_T})
\cdot\alpha_R^{c_T}({\myvec z_R}).$$
Furthermore,
$\alpha_T^{c_T}(\myvec{y_T})=\alpha_T^{c_T}(\myvec{z_T})$ (since $\myvec{y_T}=\myvec{z_T}$), thus, 
$$\alpha_1^{c_T}(y_1)\cdot\alpha_R^{c_T}({\myvec y_R}) = \alpha_1^{c_T}(z_1)\cdot\alpha_R^{c_T}({\myvec z_R}).$$
\item
We then assert, by considering all prefices of
$c_T$, that each factor of these two  multiplications is the same in both
cases and hence $\alpha_1^c(0) = \alpha_1^{c_T}(0) = \alpha_1^{c_T}(1) =
\alpha_1^c(1)$.
\end{itemize}

\paragraph{Formal proof.}
Fix any inputs $x_3, \ldots, x_n$ for the parties $p_3,\dots,p_n$.
Let $\myvec{y}$ be the input vector
$$ y_1 = 0, y_2 = 1, \mbox{and}~ y_k = x_k~\mbox{for}~3\leq k \leq n,$$
and let $\myvec{z}$ be the input vector 
$$ z_1 = 1, z_2 = 0, \mbox{and}~ z_k = x_k~\mbox{for}~3\leq k \leq n.$$
We first claim that the distribution over the views of $T$ when the protocol is executed with $\myvec y$
is the same as when the protocol is executed with $\myvec z$.
As $\SUM(\myvec y) = \SUM(\myvec z)$ and $\hat{f}$ is symmetric,
$\hat{f}(\myvec y)$ and $\hat{f}(\myvec z)$ are identically distributed.
Hence, by \clmref{v_T}, for any possible view $v_T$ of the set $T$,
it holds that $\pr\left[\View_{T}({\myvec y}) = v_T\right] = \pr\left[\View_{T}({\myvec z}) = v_T\right]$. 
Thus, since the view of $T$ contains the transcript $c_T$ of messages sent between the parties
in $T$ and the parties in $\set{p_1} \cup R$, we have that for any such possible transcript $c_T$, the probability
that the parties send messages according to $c_T$ is the same when the protocol is executed with $\myvec y$
and when the protocol is executed with $\myvec z$.
Furthermore, for any possible prefix $c'_T$ of any transcript $c_T$ of $T$, 
the probability of messages sent according to $c'_T$ when executing $\Pi$ with input $\myvec{y}$ is the same as when 
executing $\Pi$ with input $\myvec{z}$. This is true as this probability is merely the sum over the probabilities of 
all transcripts completing $c'_T$.

Without loss of generality, we can analyze the execution of the protocol as if at each round only a single message 
is sent by a single party. Let $j$ be such that $p_1$ sends a message in round $j$ and denote by
$h_j = h_{j-1},m_j$, the prefix of $c_T$ also viewed by $p_1$ (messages sent or received by $p_1$) in the first $j$ 
rounds, where $h_{j-1}$ is the history of messages viewed by $p_1$ in the first $j-1$ 
rounds, and $m_j$ is the message $p_1$ sends in round $j$, according to $c_T$. 
By the argument above, the probabilities of $h_{j-1}$ being seen by $p_1$ are the same when the protocol is executed 
with $\myvec y$ and when the protocol is executed with $\myvec z$ and the probabilities of $h_{j}$ being seen by $p_1$ 
are the same when the protocol is executed with $\myvec y$ and when the protocol is executed with $\myvec z$.
Thus, the probabilities of $p_1$ sending $m_j$
having seen message history $h_{j-1}$ are the same when $x_1 = 0$ and when $x_1 = 1$. Since the probability
of $p_1$ being consistent with a view $c_T$ (of $T$) is the product of the probabilities that it is consistent at each 
round, we have $\alpha_1^{c_T}(0) = \alpha_1^{c_T}(1)$. Let $c$ be a full transcript of the protocol, and $c_T$ be its
restriction to messages sent between parties in $T$ and parties in $\{p_1\} \cup R$. Since $p_1$ does not see
any message in $c$ that is not in $c_T$, it holds for every $x_1$ that $\alpha_1^{c}(x_1) = \alpha_1^{c_{T}}(x_1)$.
Thus, $\alpha_1^{c}(0) = \alpha_1^{c}(1)$.

Hence, we proved that for any lonely party $p_i$, and any full transcript of the protocol $c$, it holds that
$\alpha_i^{c}(0) = \alpha_i^{c}(1)$. Consider the all zero input vector and the input vector $\myvec {x}$
such that $x_i = 1$ if and only if $p_i$ is lonely. By \lemref{mult} we have that for any given full 
transcript $c$, the probability of $c$ being exchanged with $\myvec {0}$ is exactly the probability 
of $c$ being exchanged with $\myvec {x}$. Thus, if with probability at least $1/2$, when executing the protocol
with $\myvec {0}$, the exchanged transcript implies a value less than $n/4$, then with probability at least $1/2$, 
the protocol errs by at least $n/4$ when executed with $\myvec {x}$. Otherwise, with probability at least $1/2$, 
the protocol errs by at least $n/4$ when executed with $\myvec {0}$. 
\end {proof}      


\paragraph{Acknowledgments}

We thank Adam Smith, Yuval Ishai, and Cynthia Dwork for conversations related to the topic of this paper. 
This research is partially supported by the Frankel Center for Computer Science, and by the Israel Science Foundation (grant No.\ 860/06).

\bibliographystyle{plain}

\begin{thebibliography}{10}

\bibitem{BGW88}
M.~{Ben-Or}, S.~Goldwasser, and A.~Wigderson.
\newblock Completeness theorems for noncryptographic fault-tolerant distributed
  computations.
\newblock In {\em Proc. of the 20th ACM Symp. on the Theory of Computing},
  pages 1--10, 1988.

\bibitem{BDMN05}
A.~Blum, C.~Dwork, F.~McSherry, and K.~Nissim.
\newblock Practical privacy: the {SuLQ} framework.
\newblock In {\em Proc. of the twenty-fourth ACM SIGMOD-SIGACT-SIGART symposium
  on Principles of database systems}, pages 128--138, 2005.

\bibitem{BLR08}
A.~Blum, K.~Ligett, and A.~Roth.
\newblock A learning theory approach to non-interactive database privacy.
\newblock In {\em Proc. of the 40th ACM Symp. on the Theory of Computing},
  pages 609--618, 2008.

\bibitem{CCD88}
D.~Chaum, C.~Cr{\'e}peau, and I.~Damg{\aa}rd.
\newblock Multiparty unconditionally secure protocols.
\newblock In {\em Proc. of the 20th ACM Symp. on the Theory of Computing},
  pages 11--19, 1988.

\bibitem{CK93}
B.~Chor and E.~Kushilevitz.
\newblock A communication-privacy tradeoff for modular addition.
\newblock {\em Inform. Process. Lett.}, 45(4):205--210, 1993.

\bibitem{DFNT06}
I.~Damg{\aa}rd, M.~Fitzi, E.~Kiltz, J.~B. Nielsen, and T.~Toft.
\newblock Unconditionally secure constant-rounds multi-party computation for
  equality, comparison, bits and exponentiation.
\newblock In S.~Halevi and T.~Rabin, editors, {\em Proc. of the Third Theory of
  Cryptography Conference -- TCC 2006}, volume 3876 of {\em Lecture Notes in
  Computer Science}, pages 285--304. Springer-Verlag, 2006.

\bibitem{DiNi03}
I.~Dinur and K.~Nissim.
\newblock Revealing information while preserving privacy.
\newblock In {\em Proceedings of the Twenty-Second ACM SIGACT-SIGMOD-SIGART
  Symposium on Principles of Database Systems}, pages 202--210, 2003.

\bibitem{Dwo06}
C.~Dwork.
\newblock Differential privacy.
\newblock In M.~Bugliesi, B.~Preneel, V.~Sassone, and I.~Wegener, editors, {\em
  Proc. of the 33rd International Colloquium on Automata, Languages and
  Programming}, volume 4052 of {\em Lecture Notes in Computer Science}, pages
  1--12. Springer-Verlag, 2006.

\bibitem{DKMMN06}
C.~Dwork, K.~Kenthapadi, F.~McSherry, I.~Mironov, and M.~Naor.
\newblock Our data, ourselves: Privacy via distributed noise generation.
\newblock In {\em Advances in Cryptology -- EUROCRYPT 2006}, pages 486--503,
  2006.

\bibitem{DL09}
C.~Dwork and J.~Lei.
\newblock Differential privacy and robust statistics.
\newblock In {\em Proc. of the 41st ACM Symp. on the Theory of Computing},
  2009.

\bibitem{DMNS06}
C.~Dwork, F.~McSherry, K.~Nissim, and A.~Smith.
\newblock Calibrating noise to sensitivity in private data analysis.
\newblock In S.~Halevi and T.~Rabin, editors, {\em Proc. of the Third Theory of
  Cryptography Conference -- TCC 2006}, volume 3876 of {\em Lecture Notes in
  Computer Science}, pages 265--284. Springer-Verlag, 2006.

\bibitem{DMT07}
C.~Dwork, F.~McSherry, and K.~Talwar.
\newblock The price of privacy and the limits of {LP} decoding.
\newblock In {\em Proc. of the 39th ACM Symp. on the Theory of Computing},
  pages 85--94, 2007.

\bibitem{DNRRS09}
C.~Dwork, M.~Naor, O.~Reingold, G.~Rothblum, and S.~Vadhan.
\newblock On the complexity of differentially private data release: efficient
  algorithms and hardness results.
\newblock In {\em Proc. of the 41st ACM Symp. on the Theory of Computing},
  pages 381--390, 2009.

\bibitem{DwNi04}
C.~Dwork and K.~Nissim.
\newblock Privacy-preserving datamining on vertically partitioned databases.
\newblock In M.~Franklin, editor, {\em Advances in Cryptology -- CRYPTO 2004},
  volume 3152 of {\em Lecture Notes in Computer Science}, pages 528--544.
  Springer-Verlag, 2004.

\bibitem{EGS03}
A.~Evfimievski, J.~Gehrke, and R.~Srikant.
\newblock Limiting privacy breaches in privacy preserving data mining.
\newblock In {\em Proceedings of the Twenty-Second {ACM}
  {SIGMOD}-{SIGACT}-{SIGART} Symposium on Principles of Database Systems},
  pages 211--222, 2003.

\bibitem{FKKN09}
D.~Feldman, A.~Fiat, H.~Kaplan, and K.~Nissim.
\newblock Private coresets.
\newblock In {\em Proc. of the 41st ACM Symp. on the Theory of Computing},
  2009.

\bibitem{GRS09}
A.~Ghosh, T.~Roughgarden, and M.~Sundararajan.
\newblock Universally utility-maximizing privacy mechanisms.
\newblock In {\em Proc. of the 41st ACM Symp. on the Theory of Computing},
  pages 351--360, 2009.

\bibitem{GMW87}
O.~Goldreich, S.~Micali, and A.~Wigderson.
\newblock How to play any mental game.
\newblock In {\em Proc. of the 19th ACM Symp. on the Theory of Computing},
  pages 218--229, 1987.

\bibitem{KLNRS08}
S.~Kasiviswanathan, H.~K. Lee, K.~Nissim, S.~Raskhodnikova, and A.~Smith.
\newblock What can we learn privately?
\newblock In {\em Proc. of the 49th IEEE Symp. on Foundations of Computer
  Science}, pages 531--540, 2008.

\bibitem{MT07}
F.~McSherry and K.~Talwar.
\newblock Mechanism design via differential privacy.
\newblock In {\em Proc. of the 48th IEEE Symp. on Foundations of Computer
  Science}, pages 94--103, 2007.

\bibitem{MPRV09}
I.~Mironov, O.~Pandey, O.~Reingold, and S.~P. Vadhan.
\newblock Computational differential privacy.
\newblock In S.~Halevi, editor, {\em Advances in Cryptology -- CRYPTO 2009},
  volume 5677 of {\em Lecture Notes in Computer Science}, pages 126--142.
  Springer-Verlag, 2009.

\bibitem{NRS07}
K.~Nissim, S.~Raskhodnikova, and A.~Smith.
\newblock Smooth sensitivity and sampling in private data analysis.
\newblock In {\em Proc. of the 39th ACM Symp. on the Theory of Computing},
  pages 75--84, 2007.

\bibitem{RHS07}
V.~Rastogi, S.~Hong, and D.~Suciu.
\newblock The boundary between privacy and utility in data publishing.
\newblock In {\em Proc. of the 33rd International Conf. on Very Large Data
  Bases}, pages 531--542, 2007.

\bibitem{Smi08}
A.~Smith.
\newblock Efficient, differentially private point estimators.
\newblock Technical Report 0809.4794, CoRR, 2008.

\bibitem{W65}
S.~L. Warner.
\newblock Randomized response: A survey technique for eliminating evasive
  answer bias.
\newblock {\em Journal of the American Statistical Association},
  60(309):63--69, 1965.

\bibitem{Yao82a}
A.~C. Yao.
\newblock Protocols for secure computations.
\newblock In {\em Proc. of the 23th IEEE Symp. on Foundations of Computer
  Science}, pages 160--164, 1982.

\end{thebibliography}


\end{document}